\documentclass{article}
\usepackage[utf8]{inputenc}

\usepackage[margin=1in]{geometry}
\usepackage{setspace}
\usepackage{graphicx}
\usepackage[dvipsnames,table]{xcolor}
\usepackage{fancyhdr} 
\usepackage{amsmath, amsfonts, amsthm, amssymb}
\usepackage{tikz}
\usepackage{authblk}

\usepackage{hyperref}
\usepackage{bbm}
\usepackage{dsfont}
\usepackage{subfig}
\usepackage[noabbrev,nameinlink,capitalize,sort]{cleveref}
\usepackage{booktabs}

\usepackage[title]{appendix}
\allowdisplaybreaks

\providecommand{\keywords}[1]
{
  {\small	
  \textbf{\textit{Keywords---}} #1}
}

\usepackage{xpatch}
\makeatletter
\AtBeginDocument{\xpatchcmd{\cref@thmoptarg}{\thm@headpunct{.}}{\thm@headpunct{}}{}{}}
\AtBeginDocument{\xpatchcmd{\cref@thmnoarg}{\thm@headpunct{.}}{\thm@headpunct{}}{}{}}
\makeatother

\newcommand{\Rpn}{\R_+^n}

\newcommand{\bz}{\boldsymbol{z}}

\newcommand{\bvrho}{\boldsymbol{\varrho}}

\newtheorem{theorem}{Theorem}[section]

\newtheorem{definition}[theorem]{Definition}
\newtheorem{assumption}[theorem]{Assumption}
\newtheorem{proposition}[theorem]{Proposition}
\newtheorem{corollary}[theorem]{Corollary}
\newtheorem{example}[theorem]{Example}
\newtheorem{optimization}[theorem]{Optimization Problem}
\newtheorem{remark}[theorem]{Remark}

\newcommand{\KL}[2]{D_\mathrm{KL}({#1} \,\Vert\, {#2})}

\newcommand{\Id}{{\mathds{1}}}
\newcommand{\R}{{\mathds{R}}}
\newcommand{\E}{{\mathbb{E}}}
\newcommand{\Q}{{\mathbb{Q}}}
\renewcommand{\P}{{\mathbb{P}}}
\newcommand{\eps}{\varepsilon}

\newcommand{\vs}{{\varsigma}}
\newcommand{\V}{{\mathcal{V}}}

\newcommand{\vsmallskip}{\hspace{0.1em}}

\newcommand{\boldalpha}{{\boldsymbol{\alpha}}}

\newcommand{\boldeta}{{\boldsymbol{\eta}}}

\newcommand{\G}{{\mathcal{G}}}

\newcommand{\A}{{\mathcal{A}}}
\newcommand{\mfA}{{\mathfrak{A}}}
\newcommand{\mfC}{{\mathfrak{C}}}

\newcommand{\F}{{\mathcal{F}}}

\newcommand{\N}{{\mathcal{N}}}

\hypersetup{
  colorlinks   = true, 
  urlcolor     = blue, 
  linkcolor    = red, 
  citecolor   = blue 
}

\allowdisplaybreaks 

\usepackage{todonotes}
\usepackage{setspace}

\usepackage[natbib=true, style=apa]{biblatex}
\bibliography{references.bib}
\title{Optimal Robust Reinsurance with Multiple Insurers\footnote{Earlier versions have been presented at the Canadian Society of Applied and Industrial Mathematics (CAIMS) Annual Meeting 2024 (Kingston, Canada), the University of Copenhagen, the University of Waterloo (Canada), the 26th International Congress on Insurance: Mathematics and Economics (Edinburgh, Scotland), and the 2023 Statistical Society of Canada Annual Meeting (Ottawa, Canada).}}
\author[1,a]{Emma Kroell}
\author[1,b]{Sebastian Jaimungal}
\author[1,c]{Silvana M. Pesenti}

\affil[1]{Department of Statistical Sciences, University of Toronto}
\affil[a]{emma.kroell@mail.utoronto.ca}
\affil[b]{sebastian.jaimungal@utoronto.ca}
\affil[c]{silvana.pesenti@utoronto.ca}

\date{\today}

\begin{document}

\maketitle

\begin{abstract}
We study a reinsurer who faces multiple sources of model uncertainty. The reinsurer offers contracts to $n$ insurers whose claims follow compound Poisson processes representing both idiosyncratic and systemic sources of loss. As the reinsurer is uncertain about the insurers’ claim severity distributions and frequencies, they design reinsurance contracts that maximise their expected wealth subject to an entropy penalty. Insurers meanwhile seek to maximise their expected utility without ambiguity. We solve this continuous-time Stackelberg game for general reinsurance contracts and find that the reinsurer prices under a distortion of the barycentre of the insurers’ models. We apply our results to proportional reinsurance and excess-of-loss reinsurance contracts, and illustrate the solutions numerically. Furthermore, we solve the related problem where the reinsurer maximises, still under ambiguity, their expected utility and compare the solutions. \end{abstract}

\keywords{Stackelberg game, reinsurance, barycentre, Kullback-Leibler divergence, excess-of-loss reinsurance, proportional reinsurance}

\onehalfspacing

\section{Introduction}
Reinsurance is an integral part of the insurance industry, allowing smaller, first-line insurers to take on a larger client base while passing on large risks to reinsurance companies. Therefore, reinsurance companies play an important stabilising role in financial systems --- but they remain vulnerable to events that cause large, industry-wide losses, such as natural disasters. The cost of these events is increasing, with 100B USD in insured losses from natural catastrophes in 2022.\footnote{\href{https://www.swissre.com/press-release/Insured-losses-from-natural-catastrophes-break-through-USD-100-billion-threshold-again-in-2022/e74c6ce7-8914-45d6-a384-df71dbeb87b8}{https://www.swissre.com/press-release/Insured-losses-from-natural-catastrophes-break-through-USD-100-billion-threshold-again-in-2022/e74c6ce7-8914-45d6-a384-df71dbeb87b8}} Furthermore, while the reinsurance industry is large, it is also concentrated ---  reinsurance premia of the top 50 reinsurance players totalled 353B USD in 2021, with 24\% underwritten by the top two.\footnote{\href{https://www.insurancejournal.com/news/international/2022/08/17/680847.htm}{https://www.insurancejournal.com/news/international/2022/08/17/680847.htm}} Given these facts, we examine the pricing and risk-management behaviour of a large reinsurance company who contracts with many different first-line insurance companies but exhibits different degrees of ambiguity for each insurance company. The insurers are insuring similar risks but have different models for losses. We approach this problem as a Stackelberg game and examine the optimal behaviour of the large reinsurer and the individual insurance companies. 

The problem of optimal reinsurance has been studied extensively in the one-period setting, dating back to \textcite{borch1960reciprocal}. 
We focus on optimal reinsurance in continuous time, where the literature has typically taken the insurer's perspective. Examples include \textcite{hojgaard1998optimal}, \textcite{Irgens2004}, and \textcite{Liang2011} who look at an insurer purchasing reinsurance to maximise their wealth or utility and  \textcite{schmidli2001optimal} and \textcite{schmidli2002minimizing} who consider an insurer minimising their probability of ruin. Recent studies allowed for ambiguity-aversion of the insurer (\cite{Yi2013}; \cite{Li2018}). However, these studies leave open the question of whether the reinsurer would accept the individual insurer's optimal contract.
Recent years have seen considerable work aimed at closing this gap by modelling the behaviour of the reinsurer as a continuous-time Stackelberg game. This was first introduced by \textcite{ChenShen2018}, who solve the optimal proportional reinsurance problem when the insurer and reinsurer's surpluses follow diffusion processes and both agents maximise exponential utility functions. \textcite{ChenShen2019} consider the mean-variance problem for a general reinsurance treaty, and find that proportional reinsurance is optimal under a variance premium while excess-of-loss reinsurance is optimal with an expected value premium. 

Ambiguity-aversion in the Stackelberg game setting was introduced by \textcite{Hu2018continuoustime} for proportional reinsurance contracts, where the insurer's process is approximated by a diffusion. They allow for uncertainty in the drift of the insurer's claims process, and show that the ambiguity-averse reinsurer expects higher claims and sets higher reinsurance prices, which decreases the demand for reinsurance. In the Cram\'er-Lundberg setting, \textcite{Hu2018} allow for ambiguity of the reinsurer about the intensity of the insurer's compound Poisson process, and show that the ambiguity-averse reinsurer expects a higher intensity of the compound Poisson process. Furthermore, for exponential claim distributions, they show that this causes the reinsurer to charge a higher premium, and the insurer to purchase less reinsurance for both proportional and excess-of-loss reinsurance contracts. The framework is further extended by \textcite{Gu2020}, who study the excess-of-loss reinsurance problem and allow for investment in a risky asset by both the insurer and reinsurer, and show that the behaviour of the agents at equilibrium depends on each other's ambiguity aversions. Most recently, \textcite{Cao2022} and \textcite{Cao2022generaldivergence} study the related expected wealth reinsurance problem with squared deviation and general divergence ambiguity penalties, respectively. They show that the insurer's and reinsurer's probability distortions are increasing in the loss, and that proportional reinsurance is optimal under squared-error divergence, but not necessarily under a general divergence measure.

While the optimal reinsurance problem with more than a single insurer or reinsurer has been well-studied in the one-period setting (see, e.g., \textcite{boonen2016role}, \textcite{BoonenGhossoub2019}, and \textcite{BoonenGhossoub2021}), to our knowledge, its study in continuous time is limited. An exception is \textcite{Cao2023treeVchain} and \textcite{Cao_Li_Young_Zou_2023}, who consider the reinsurance problem with multiple reinsurers and different configurations of reinsurance agreements. \textcite{Bai2022} consider the case of two competing insurers in a hybrid Stackelberg game, who maximise their expected wealth with delay and show that their competition reduces their demand for reinsurance.

This work is distinct from the aforementioned literature in several ways. First, our focus is on the behaviour of an ambiguity-averse reinsurer who takes on potentially different contracts from \emph{multiple} insurers, each having different loss models. The reinsurer allows for distortions to both the intensity and severity of the loss distributions, and moreover accounts for different degrees of ambiguity for each insurer's loss models using a weighted Kullback-Leibler divergence penalty. Second, the losses faced by the different insurers consist of idiosyncratic losses, characteristic to each insurer individually, and systemic losses that affect all insurers simultaneously.
Third, rather than considering a specific reinsurance contract, we solve the Stackelberg game for a general reinsurance contract, encompassing proportional and excess-of-loss reinsurance. Our reinsurance contracts take the form of a \emph{retention function} $r(z, \alpha_k)$, where $z$ is the loss and $\alpha_k$ is an aspect of the contract chosen by the insurer.\footnote{For an overview of reinsurance forms, see, e.g., \textcite[Chapter 2]{Albrecher2017} or \textcite[Chapter I.4]{Asmussen2020}.} We consider a broad class of retention functions which include most common reinsurance contracts, including proportional reinsurance and excess-of-loss reinsurance, with or without a policy limit. Fourth, we compare the wealth- and exponential utility- maximisation cases, bringing together these two distinct approaches in the literature. 

We derive the Stackelberg equilibrium for a general reinsurance contract, and find that the reinsurer's optimal pricing model is closely related to the barycentre of the insurers' models. We show that the optimal reinsurance prices, demand, and the reinsurer's pricing model are independent of time. Furthermore, we show that when the reinsurer is ambiguity-averse, their optimal pricing model depends on the size of their aggregate loss. We further consider the Stackelberg equilibrium for specific reinsurance contracts --- proportional reinsurance and excess-of-loss reinsurance, with or without a policy limit --- and illustrate the solution numerically in the case of two insurers. Finally, we solve the related problem of an ambiguity-averse reinsurer who maximises a CARA utility and compare the solutions, thus generalising known results in the literature about utility-maximising reinsurers. We find that the safety loadings charged to the insurers are increasing in the reinsurer's risk aversion parameter.

The rest of the paper is structured as follows. \cref{sec:ins_mrkt} details the wealth evolution of the insurers and the reinsurer and defines the general reinsurance contract. \cref{sec:ins_prob} provides the solution to the insurer's problem. \cref{sec:reins_prob} details and solves the reinsurer's optimisation problem, and presents the Stackelberg equilibrium.
In \cref{sec:application}, we examine the Stackelberg equilibrium for different reinsurance contract types, namely, proportional reinsurance and excess-of-loss reinsurance with and without a policy limit. We provide numerical illustrations of the equilibrium and its key features for the case of two insurers with given severity distributions. Finally, in \cref{sec:utility_max} we solve the related utility-maximisation problem for the reinsurer under ambiguity, and compare it to the wealth-maximisation approach.

\section{Reinsurance market}
\label{sec:ins_mrkt}

We consider a reinsurance market with $n \in \mathbb{N}$ non-life insurance companies and a single reinsurer over a finite time horizon $[0,T]$. The time $T$ is the expiration date of the reinsurance contracts, which are entered into at time 0. Each insurer provides non-life insurance policies to policyholders and receives regular premium payments. They further aim to transfer a part of their potential losses to a reinsurance company. Each insurer decides their optimal level of reinsurance based on the reinsurance price and the contract terms, with the goal of maximising their expected utility. Furthermore, each insurer has their own model of the insurable losses in the market, which they use to estimate their losses and optimise their behaviour. On the other hand, the reinsurer determines the prices at which they offer reinsurance to each individual insurer, aiming to optimise their expected wealth subject to model uncertainty. Given the $n$ different models for the insurable losses (from the $n$ insurers), the reinsurer has to determine which model to use for pricing, accounting for the differing models and potential ambiguity about them. We formalise this as a Stackelberg game, where the reinsurer is the leader and the insurers are the followers.

\subsection{The probabilistic setup}

We consider a complete and filtered measurable space $(\Omega, \F, \mathbb{F}=(\F_t)_{t \in [0,T]})$ and $n$ equivalent probability measures $\P_1, \ldots, \P_n$. For $k \in \N:=\{1, \ldots, n\}$, the probability measure $\P_k$ encapsulates the $k$-th insurer's beliefs about the loss distributions. We  refer to an individual insurer as insurer-$k$ or the $k$-th insurer, $k \in \N$. 
Throughout, unless stated otherwise, we assume the filtration $\mathbb{F}$ is such that all processes defined below are progressively measurable with respect to it. 

There are two types of loss events. Either an idiosyncratic loss event occurs that affects only a single insurer, having no effect on the others, or a systemic loss event occurs, affecting all $n$ insurers. This distinction is important for the reinsurer, who needs to model all losses to appropriately price reinsurance contracts. In particular, they may use different modelling assumptions for the idiosyncratic and systemic losses.  

To formulate this mathematically, let $M_1(dz_1,dt),\ldots,M_n(dz_n,dt)$ be $n$ independent Poisson random measures (PRMs). Each PRM is associated with a single insurer, describing their idiosyncratic (individual) losses. Under the probability measure $\P_j$, $j\in\N$, the compensator of $M_k(dz_k,dt)$, $k \in \N$, is denoted by $\mu_{k,j} (dz_k,dt)$, i.e., $\int_0^t \int_{\R_+} z_k \left[M_k(dz_k,du)-\mu_{k,j}(dz_k,du)\right]$ is a $\P_j$-martingale, where $\R_+: = (0, \infty)$.
We assume that each insurer's losses are compound Poisson, so we may write 
 \begin{equation}
 \label{eqn:indep_comp}
     \mu_{k,j} (dz_k,dt) = \mu_{k,j} (dz_k) \, dt = \lambda_{k,j}^I \, F_{k,j}^I(dz_k) \, dt \,,
 \end{equation}
where $F_{k,j}^I$ is the cumulative distribution function of insurer-$k$'s idiosyncratic losses with support on $\R_+$ and $\lambda_{k,j}^I\in \R_+$ is the arrival rate of insurer-$k$'s individual losses under $\P_j$. We assume the compensator admits a density $w_{k,j}(z_k)$ such that $\mu_{k,j}(dz_k,dt) = w_{k,j}(z_k)\, dz_k\,dt$.

Furthermore, let the PRM $N(d\bz,dt)$, $\bz = (z_1, \ldots , z_n)$, be the source of systemic (shared) risk, assumed to be independent of $M_k$, for all $k\in\N$. Throughout, we use bold text to denote vectors. Under the probability measure $\P_j$, $j\in\N$, the compensator of $N(d\bz,dt)$ is denoted by $\nu_{j}(d\bz,dt)$, i.e., $\int_0^t \int_{\Rpn} \bz \big[N(d\bz,du) -\nu_{j}(d\bz,du)\big]$ is a $\P_j$-martingale. We further assume that the systemic loss is compound Poisson, i.e.,
\begin{equation}
    \nu_{j}(d\bz,dt) = \lambda_{j}^S \,G_j(d\bz) \,,
\end{equation}
where, under $\P_j$, $j\in\N$, $\lambda^S_{j}\in \R_+$ is the arrival rate of the systemic loss and $G_{j}$ is the joint cumulative distribution function of the systemic losses with support on $\Rpn$.  We assume the compensator admits a density $v_{j}(\bz)$ such that $\nu_{j}(d\bz,dt) = v_{j}(\bz)\, d\bz\,dt$.

As the systemic process generates marginal losses for each insurer, we further denote the insurer-$k$'s marginal compensator of the systemic loss under $\P_j$, $j\in\N$, by $\nu_{k,j}$, i.e.,
\begin{align*}
    \nu_{k,j}(dz_k)dt = \nu_j(\R_+, \ldots, \R_+, dz_k, \R_+, \ldots, \R_+) \,dt 
    =\lambda_j^S \int_{\R^{n-1}_+} G_j(d\bz) \,dt
    = \lambda^S_{j} \, F^S_{k,j}(dz_k) \, dt\,,
\end{align*}
where $F^S_{k,j}$ denotes the marginal distribution of the systemic loss under $\P_j$, $j\in\N$. Insurer-$k$'s losses are thus the sum of their idiosyncratic losses and the marginal losses resulting from the systemic process.

\subsection{Insurers' wealth evolution}

Given this probabilistic setup, we next describe the individual insurers' losses. Insurer-$k$ calculates losses using the probability measure $\P_k$, which summarises their views on the insurable losses. We therefore derive a concise representation for their losses under this probability measure. 

As discussed above, they face two sources of loss. The first source of loss is the claims stemming from idiosyncratic events that only they face, which arrive according to a Poisson process with constant intensity $\lambda_{k,k}^I$ and claim severity distribution function $F^I_{k, k}(\cdot)$. The second source is the claims stemming from the systemic event, shared with all other insurers, and which arrive with constant intensity $\lambda^S_k$ and claim severity distribution function $F^S_{k, k}(\cdot)$. Therefore, each insurer's loss process follows a Cram\'{e}r-Lundberg model, where, under their probability measure $\P_k$, claims arrive with intensity $\lambda_k$ and the claim severity distribution function is $F_k$, where
\begin{equation}
\label{eqn:ins_cPois}
    \lambda_k := \lambda_k^S + \lambda_k^I, \qquad 
    F_k(dz_k) := \frac{\lambda_k^S}{\lambda_{k}} F_{k,k}^S(dz_k) + \frac{\lambda_{k,k}^I}{\lambda_k} F_{k,k}^I(dz_k) \,.
\end{equation}

Each insurer sets a constant premium rate $p_k^I$; for example the expected value principle with fixed safety loading.
Each insurer then purchases reinsurance on their aggregate claim amount. The available reinsurance contracts are characterised by a retention function $r$. Insurer-$k$ chooses an aspect of the reinsurance contract $a_k \in \A\subseteq \R$, which we refer to as insurer-$k$'s control. The interpretation of $a_k$ and the set of possible controls $\A$ will vary depending on the form of the retention function, as shown in the examples below. For a choice of retention function $r$ and insurer-$k$'s control $a_k$, the reinsurer charges a reinsurance premium $p^R$ with safety loading $c_k$.

\begin{definition}[Reinsurance Contract]
A reinsurance contract is characterised by a retention function $r\colon [0,\infty) \times \A \to [0,\infty)$ and a corresponding reinsurance premium $p^R\colon \A\times \R_+ \to \R_+$. The retention function is non-decreasing in the first argument, satisfying $0 < r(z, a) \leq z$ for all $z \in \R_+$, $a \in \A$ and $r(0,a)=0$ for any $a\in\A$.

For a tuple $(a, c) \in \A \times \R_+$, the reinsurer agrees to cover $z-r(z, a)$ for a premium $p^R (a, c)$, where $c$ is the reinsurer's safety loading.
\end{definition}

The bounds on the retention function guarantee that the insurer must retain some of the losses. Typically the inequality, $r(z,a)\le z$, is strict for some $z\in\R_+$.

\begin{assumption}
\label{assump:retentionfuncs}
    We consider retention functions that are continuous in the loss $z$ and increasing and almost everywhere differentiable in the control $a$.
\end{assumption}

Of particular interest are piecewise linear retention functions, as these are the reinsurance contracts encountered in practice. Examples of such retention functions include the following:

\begin{example}
\label{ex:prop}
Proportional reinsurance: $r(z, a) = a \vsmallskip z$, $a \in (0,1]$. Here, the insurer chooses the proportion of the loss to retain.
\end{example}

\begin{example}
\label{ex:XL}
Excess-of-loss insurance: $r(z, a) = \min \{ a, \,z \}$, $a \in \R_+$. Here, the insurer chooses the retention limit, beyond which the reinsurer covers any excess losses.
\end{example}

\begin{example}
\label{ex:XL_capped}
Excess-of-loss insurance with a fixed policy limit $\ell$: $r(z, a) = \min \{a, z \}\, \Id_{z \leq a + \ell} + (z - \ell)\,\Id_{z > a + \ell} $, $a \in \R_+$, $\ell > 0$. Here, the insurer chooses the retention limit, beyond which the reinsurer covers any excess losses. If however the losses exceed the limit $a+\ell$, the losses are borne by the insurer. We assume that the insurer cannot choose $\ell$ as the policy limit is in practice set by the reinsurer, e.g. by the purchase of further reinsurance (see, e.g., \textcite[Chapter 3.4]{Mikosch2009}).
\end{example}

\begin{remark}
Instead of considering retention functions of the form $r(z,a)$ with control variable $a$, one could alternatively place no restrictions on the functional form of the retention function other than $0 \leq r(z) \leq z$, and allow the insurer to optimise over this class of retention functions. This approach is taken by \textcite{Li2017}, who show that in this case the optimal contract is an excess-of-loss contract. In our setting, as the reinsurer may prefer to offer other types of contracts, such as proportional reinsurance, we do not make this assumption.
\end{remark}

The insurer's control of the retention function may vary over time, thus we denote by $\alpha_k := (\alpha_{t,k})_{t\in [0,T]}$ the insurer-$k$'s control process. Further, we denote by $\mfA$ the set of square-integrable, $\mathbb{F}$-predictable processes taking values in $\A$. We require that $\alpha_k\in\mfA$. 
Thus, the $k$-th insurer's wealth process $X_k := (X_{t,k})_{t\in[0,T]}$ evolves according to the stochastic differential equation (SDE)
\begin{equation}
\label{eq:insurer_wealth}
    X_{t,k} = X_{0,k} + \int_0^t\!\! \left[ p_k^I - p^R_{k}(\alpha_{u,k}, c_k) \right] du - \int_0^t \int_{\R_+} \!\!r(z_{k},\alpha_{u,k}) \,M_{k}(dz_k,du) - \int_0^t \int_{\Rpn} \!\!r(z_{k},\alpha_{u,k}) \,N(d \bz,du) \,,
\end{equation}
where $p_k^R$ is the reinsurance premium for insurer-$k$ with safety loading $c_k \in\R_+$.
The first integral is the difference of the premia the insurer collects from their policyholders and the premium they pay to the reinsurer. The second integral is the losses retained by the insurer from their idiosyncratic loss process, while the third integral is the losses retained by the insurer from the systemic loss process. 

We apply the retention function $r$ to each loss source separately (rather than jointly) since $M_k$ and $N$ are independent, and hence never generate losses at the same time. In fact, from insurer-$k$'s perspective, the second and third terms in \eqref{eq:insurer_wealth} could be treated as a single loss process, with compensator $\lambda_k F_k(dz_k)$, as given in \eqref{eqn:ins_cPois}, and thus be modelled as one process.

\subsection{Reinsurer's wealth evolution}

The reinsurer offers reinsurance contracts to each insurer, where they cover the losses not retained by insurer-$k$, i.e., $z - r(z,a_k)$, and charge a continuous premium rate $p^R_{k}(a_k,c_k)$. As the reinsurer knows how each insurer will respond to the price, their premium depends on the insurer's control $\alpha_k$. We assume that the reinsurer sets the reinsurance premium rate for insurer-$k$ using the expected value principle with deterministic safety loading $\eta_k: =(\eta_{t,k})_{t\in [0,T]}$. Following the convention in the literature (see, e.g., \textcite{Hu2018}, \textcite{Gu2020}, and \textcite{Cao2022}), we treat the safety loading $\eta_k$ as the price set by the reinsurer, to be optimised under the reinsurer's probability measure, while the premium itself is expressed under the insurer's probability measure $\P_k$. Therefore, the premium charged to insurer-$k$ is 
\begin{equation}
\label{eqn:reins_prem}
\begin{split}
    p^R_{t,k} = p^R_{t,k}(\alpha_k,\eta_{t,k}) &:=(1+\eta_{t,k}) \left[ \int_{\R_+} \!\!\! \left[ z_k - r(z_k,\alpha_{t,k} )\right] \mu_{k,k}(dz_{k}) + \int_{\Rpn} \!\!\! \left[ z_{k} - r(z_k,\alpha_{t,k} )\right] \nu_k(d\bz) \right] \\
    &=(1+\eta_{t,k}) \, \lambda_k \int_{\R_+} \!\!\! \left[ z_k- r(z_k,\alpha_{t,k} )\right] F_k(dz_k)\, .
\end{split}
\end{equation}

Let $\boldeta := (\eta_1,\ldots,\eta_n)$. We assume each component of $\boldeta$ is a non-negative square integrable deterministic process and denote by $\mfC$ the space of such $\boldeta$. Summing over the premia charged to each insurer and the losses taken on from each insurer, the reinsurer's wealth process $Y := (Y_t)_{t\in[0,T]}$ evolves according to the SDE
\begin{equation*}
    \begin{split}
        Y_t =Y_0 &+ \sum_{k\in\N}\int_0^t  p^R_{k}(\alpha_{u,k},\eta_{u,k}) \, du \, -  \sum_{k\in\N} \int_0^t \!\! \int_{\R_+} \!\! \big[z_{k}-r(z_{k},\alpha_{u,k})\big] M_{k}(dz_k,du)\\
        &- \int_0^t \int_{\Rpn} \sum_{k\in\N} \, \big[z_{k}-r(z_{k},\alpha_{u,k})\big] N(d\bz,du) \, .
    \end{split}
\end{equation*}

With the evolution of the wealth of the insurers and reinsurer at hand, we next discuss the optimal strategy of each individual insurer, who maximises their utility.  

\section{Insurers' utility maximisation strategy}
\label{sec:ins_prob}

The problem of optimal reinsurance from the insurer's perspective has been well studied in the literature for particular retention functions and various objectives for the insurer. While a key focus of this work is on the behaviour of the reinsurer, in this section we present the solution to the insurers' optimal reinsurance problem in our framework, for a fixed but unspecified retention function. The insurers' optimal demand for different reinsurance contracts given a fixed price will be key to solving the reinsurer's optimal pricing problem.

In our setting, each insurer seeks to maximise their expected utility, given their beliefs about the losses. We assume that insurer-$k$ has an exponential, or constant absolute risk aversion (CARA), utility function $u_k(x) = -\frac{1}{\gamma_k} e^{-\gamma_k x}$ with risk aversion parameter $\gamma_k > 0$.

\begin{optimization}
Insurer-$k$ seeks the contract parameter that attains the supremum
\begin{equation}
    \label{eqn:P-I}
    \sup_{\alpha_k\in\mfA} \E^{\P_k} \left[\, u_k(X_{T,k}) \,\right]\,,
    \tag{$P_I$}
\end{equation}
where $\E^{\P_k}[\cdot]$ denotes the expected value under $\P_k$.
\end{optimization}

The proposition below characterises the optimal choice of the insurer's control $\alpha_k$. The condition may appear strong, however, as far as we are aware, it is satisfied by all examples in the extant literature. 

\begin{proposition}
    \label{prop:insurer_optim}
    For $k\in\N$, $c \in \R_+$, and a retention function $r$, consider the following non-linear equation for $a\in\A$
    \begin{equation}
        \label{eq:insurer_optim_alpha}
        \int_{\R_+} \!\! \partial_{a}r(z_k,a) \left\{(1+c)-e^{\gamma_k r(z_k,a)}\right\} F_k(dz_k)  = 0 \, 
    \end{equation}
and denote by $\alpha_k^\dagger[c]$ its solution, if it exists. If furthermore, 
    \begin{equation}
    \label{assump:convexinsurer}
        \int_{\R_+}  \!\!  \partial_a^2 r(z_k,a)\Big|_{a=\alpha_k^\dagger[c]}
        \left( (1 + c) - e^{\gamma_k r(z_k,\alpha_k^\dagger[c])} \right) F_k(dz_k)  \leq 0\,,
        \quad \forall\;  t\in[0,T]\,,
    \end{equation}
then the process $\alpha_{t,k}^*:=\alpha^\dagger_k[\eta_{t,k}]$, $t\in[0,T]$, is insurer-$k$'s optimal control in feedback form of problem \eqref{eqn:P-I}.
\end{proposition}

We can say more about the existence of a solution to \eqref{eq:insurer_optim_alpha} for specific choices of the retention function. In particular, we derive simpler expressions for commonly used retention functions considered in the literature, including those of proportional and excess-of-loss reinsurance contracts, as described below.

\begin{proof}
Insurer-$k$'s value function is
\begin{equation*}
    J_k(t,x) := \sup_{\alpha_k\in\mfA} \E^{\P_k}_{t,x} \left[ -\tfrac{1}{\gamma_k}e^{-\gamma_kX_{T,k}} \right] ,
\end{equation*}
where $\E^{\P_k}_{t,x}$ denotes the conditional $\P_k$-expectation given that the insurer's wealth process satisfies $X_{t-,k}=x$. The dynamic programming principle gives that the value function satisfies the following Hamilton-Jacobi-Bellman (HJB) equation 
\begin{alignat*}{2}
    0 &= \partial_t J_k(t,x) + \sup_{a_k\in\A} \Bigg\{ \Big[ p_k^I- && (1+c) \, \lambda_k \int_{\R_+} \!\!\! \left[ z_k - r(z_k,a_{k} )\right] F_k(dz_k)  \Big]\, \partial_x J_k(t,x)  
    \\
    &  && + \int_{\R_+} \!\! \left[J_k(t,x-r(z_{k},a_k)) - J_k(t,x)\right]\, \mu_{k,k}(dz_{k},dt)\\
    &  &&+ \int_{\Rpn} \! \left[J_k(t,x-r(z_k,a_k)) - J_k(t,x)\right]\, \nu_{k}(d\bz,dt) \Bigg\} ,
    \\
    J_k(T,x) &= -\tfrac{1}{\gamma_k}\,
    e^{-\gamma_k\,x}.
\end{alignat*}
Given the functional form of the terminal condition, we make the Ansatz
\begin{equation}\label{eq:ansatz}
    J_k(t,x) = -\tfrac{1}{\gamma_k}\,e^{-\gamma_k\,[f(t) + g(t) x]}\, ,
\end{equation}
for deterministic functions $f(t)$ and $g(t)$ satisfying terminal conditions $ f(T) = 0$ and $g(T) = 1$, respectively. Upon substituting \eqref{eq:ansatz} into the HJB equation, and after simplifying, we obtain
\begin{equation}
\begin{aligned}
\label{eq:insurerPDE}
    0 = \gamma_k \left[ f'(t) + g'(t) \, x \right] + \sup_{a_k\in\A} \Bigg\{  \gamma_k \,g(t)\Big[ & p_k^I - (1+c) \lambda_k \int_{\R_+} \!\!\! \left[ z_k - r(z_k,a_{k} )\right] F_k(dz_k)   \Big]
    \\
    &  -\lambda_k \int_{\R_+}  \!\!  \left[ e^{\gamma_k g(t) r(z_k,a_k)} - 1\right] F_k(dz_k)  \Bigg\} \, .
\end{aligned}
\end{equation}
As the above must hold for all $x$, it follows that $g(t)\equiv 1$, for all $t\in[0,T]$.
The first order conditions (in $a_k$) imply that the supremum satisfies \cref{eq:insurer_optim_alpha}. Furthermore the condition in  \cref{assump:convexinsurer} ensures that $(\alpha_{t,k}^*)_{t\in[0,T]}$ is indeed a local maximiser. 
\end{proof}

\cref{prop:insurer_optim} implies that the insurer reduces their demand for reinsurance (retains more of the aggregate loss) as the price increases.
\begin{corollary}
    The insurer's optimal response to the reinsurer's price, $\alpha_k^\dagger[c]$, is an increasing function of the price $c$.
\end{corollary}
\begin{proof}
    Taking the derivative of \cref{eq:insurer_optim_alpha} gives that 
        \begin{align*}
        \frac{\partial \alpha_k}{\partial c} &= \frac{\displaystyle \int_{\R_+}  \!\! \partial_{a}r(z_k,a)F_k(dz_k)}{\displaystyle  \int_{\R_+}  \!\!\!  \left( \gamma_k \left( \partial_{a}r(z_k,a) \right)^2  e^{\gamma_k r(z_k,a)} - \partial_a^2 r(z_k,a)
        \left[ (1 + c) - e^{\gamma_k r(z_k,a)} \right]\right) F_k(dz_k)},
    \end{align*}
    which is positive by \cref{assump:convexinsurer} and the fact that the retention function $r$ is non-decreasing in $a$.
\end{proof}

Next, we examine the insurer's optimal response when the reinsurance contract is of a specific type. If the insurer chooses the retention level of an excess-of-loss reinsurance contract, then \cref{eq:insurer_optim_alpha} in \cref{prop:insurer_optim} simplifies to an explicit expression for $\alpha^*_k$. This is in contrast to proportional reinsurance, where this is not the case. 

When \cref{prop:insurer_optim} is applied to the specific cases of proportional reinsurance and excess-of-loss reinsurance, discussed next, we recover known results in the literature, such as those of \textcite{Li2017} and \textcite{Hu2018}.

\begin{example}[Proportional reinsurance]
In the case of proportional reinsurance (\cref{ex:prop}), the non-linear equation for $a\in\A$ in \cref{prop:insurer_optim} is
\begin{equation*}
    (1+c) \int_{\R_+}  \!\! z_k \, F_k(dz_k) = \int_{\R_+}  \!\!  z_k e^{\gamma_k a z_k} F_k(dz_k) \, .
\end{equation*}
This simplifies to
\begin{equation}\label{example:prop-reinsurance-insurer-control}
    (1+c) \, \E^{\P_k} \! \left[ Z \right] = \E^{\P_k} \! \left[ Z\, e^{\gamma_k a Z} \right]\, ,
\end{equation}
where $Z$ is a random variable with $\P_k$-distribution function given in \eqref{eqn:ins_cPois}.
If a solution $\alpha_k^\dagger[c]$ to \cref{example:prop-reinsurance-insurer-control} exists, then it is insurer-$k$'s optimal control. Furthermore, it is a global maximiser, since for proportional reinsurance, the expression in the curly braces in \cref{eq:insurerPDE} is strictly convex in $a$.
\end{example}

\begin{example}[Excess-of-loss reinsurance]
In the case of excess-of-loss reinsurance (\cref{ex:XL}), we have $\partial_{a}r(z,a) = \Id_{\{a \leq z\}}$, therefore
the condition in \cref{prop:insurer_optim} reduces to
\begin{equation*}
    (1+c) \int_{a}^\infty \! F_k(dz_k)  = \int_{a}^\infty  e^{\gamma_k  a} \,  F_k(dz_k) \, .
\end{equation*}
Rearranging gives
\begin{equation}
    \alpha_k^\dagger[c] = \tfrac{1}{\gamma_k}\log(1+c)\,.
\end{equation}

Furthermore, $\alpha_k^\dagger[c]$ is a global maximiser as the derivative of \eqref{eq:insurerPDE} with respect to $a_k$ is increasing on $(0, \alpha_k^\dagger[c])$, decreasing on $(\alpha_k^\dagger[c], \overline{\alpha})$, and constant on $(\overline{\alpha},\infty)$, where $\overline{\alpha}$ is the essential supremum of insurer-$k$'s claim severity distribution.
\end{example}

We note that when the insurer seeks to maximise their exponential utility over all possible retention functions (without fixing a specific contract type), the optimal contract form is an excess-of-loss contract \citep{Li2017}, so this contract form is of particular interest.

We further consider excess-of-loss reinsurance with a policy limit, and find that the limit does not affect the insurer's optimal choice of retention level. 

\begin{example}[Excess-of-loss reinsurance with a policy limit]
In the case of excess-of-loss reinsurance with a policy limit (\cref{ex:XL_capped}), a similar calculation to the previous example shows that
the condition in \cref{prop:insurer_optim} simplifies to
\begin{equation}
    \alpha_k^\dagger[c] = \tfrac{1}{\gamma_k}\log(1+c)\, ,
\end{equation}
and  $\alpha_k^\dagger[c]$ is a global maximiser.
\end{example}

We observe in all three examples that $\alpha_k^\dagger$ is an increasing function of $c$, the reinsurer's safety loading. Recall that $\alpha^\dagger_k$ is in \cref{ex:prop} the percentage of the loss the insurer retains, and in \cref{ex:XL,ex:XL_capped} the level at which the insurer purchases reinsurance. Thus $\alpha_k^\dagger$ being increasing in the reinsurer's safety loading implies that the insurer purchases less reinsurance if the safety loading increases.

\section{Reinsurer's strategy under ambiguity}
\label{sec:reins_prob}

Next, we consider the reinsurance problem from the perspective of the reinsurer. We assume that, unlike the insurer, the reinsurer does not have a utility but rather seeks to maximise their expected wealth. They are, however, ambiguity averse in the following sense: as their premia and losses depend on the insurers' loss models, the reinsurer must calculate reinsurance premia that take into account the different insurers' loss models, as well as uncertainty about their accuracy. Using the insurers' models $\{\P_k\}_{k \in \N}$ as reference measures, the reinsurer seeks to find a single probability measure to price reinsurance contracts and estimate losses.

Specifically, the reinsurer seeks to maximise their wealth under a single probability measure $\Q$ that accounts for the different insurers' loss models as well as uncertainty about their accuracy. The reinsurer chooses $\Q$ from the set of measures that are absolutely continuous with respect to the reference measures by choosing the compensators of the PRMs under $\Q$, satisfying the following conditions.

\begin{definition}[Reinsurer's compensator]
An admissible collection of compensators for the reinsurer is a collection of nonnegative, $\mathbb{F}$-predictable random fields 
$\boldsymbol{\sigma} = (\varsigma, \bvrho)$, where
$\varsigma=(\varsigma_t(\cdot))_{t \in [0,T]}$ and where $\bvrho=(\varrho_{t,1}(\cdot),\ldots,\varrho_{t,n}(\cdot))_{t \in [0,T]}$ such that, for all $t\in[0,T]$, $\varsigma_t(\cdot):\Rpn\to\R_+$ and $\varrho_{t,k}(\cdot):\R_+\to\R_+$, for $k \in \N$. Furthermore, assume
\begin{equation*}
    \E^{\P_j} \left[ \exp \left( \, \sum_{k\in\N}\int_0^T \!\! \int_{\R_+}   \left[ 1-\frac{\varrho_{t,k}(z_k)}{w_{k,j}(z_k)} \right]^2 M_{k}(dz_k, dt) + \int_0^T \!\! \int_{\Rpn}  \left[1- \frac{\varsigma_{t}(\bz)}{v_{j}(\bz)} \right]^2 N (d\bz, dt) \right) \right] < \infty , \quad \forall j \in \N \,,
\end{equation*}
where $w_{k,j}(z_k)$ is the density of insurer-$k$'s idiosyncratic loss PRM $\P_{j}$-compensator and $v_j(\bz)$ is the density of the systemic loss PRM's $\P_{j}$-compensator.
We denote by $\mathcal{V}$ the set of admissible collections $\boldsymbol{\sigma}$.
\end{definition}

Next, $\boldsymbol{\sigma} \in \V$ induces a probability measure $\Q^{\boldsymbol{\sigma}}$, which can be characterised by the Radon-Nikodym derivative from $\P_j$ to $\Q^{\boldsymbol{\sigma}}$, $j \in \N$, 
\begin{align*}
    \frac{d\Q^{\boldsymbol{\sigma}}}{d\P_j} &= \exp \Bigg( \, \sum_{k\in\N} \int_0^T \!\! \int_{\R_+} \log \left( \frac{\varrho_{t,k}(z_k)}{w_{k,j}(z_k)}\right) M_{k}(dz_k,dt) + \sum_{k\in\N} \int_0^T \!\!\int_{\R_+}  \left[1- \frac{\varrho_{t,k}(z_k)}{w_{k,j}(z_k)} \right] w_{k,j}(z_k) \,dz_k \,dt \\
    & \qquad\qquad + \int_0^T \!\! \int_{\Rpn} \log \left( \frac{\varsigma_{t}(\bz)}{v_{j}(\bz)}\right) N(d\bz,dt) +  \int_0^T \!\! \int_{\Rpn} \left[1- \frac{\varsigma_{t}(\bz)}{v_{j}(\bz)} \right] v_{j}(\bz) \,d\bz \,dt \Bigg) \, .
\end{align*}
Then by Girsanov's theorem (see, e.g., \textcite{appliedstochjump}), the systemic loss process $N(d\bz,dt)$ has $\Q^{\boldsymbol{\sigma}}$-compensator $\nu^{\Q^{\boldsymbol{\sigma}}}(d\bz,dt) := \vs_t(\bz) \, d\bz \, dt$ and the insurers' individual loss processes $M_k(dz_k,dt)$ have $\Q^{\boldsymbol{\sigma}}$-compensators $\mu_k^{\Q^{\boldsymbol{\sigma}}}(dz_k,dt) := \varrho_{t,k}(z_k) \, dz_k\, dt$. For future use, we define the set $\G$ as the set consisting of functions $g:[0,T] \times \Rpn \to\R_+$ such that $\int_{\Rpn} g(t,\bz)\,d\bz < \infty$ for every $t \in [0,T]$ and the set $\mathcal{H}$ as the set consisting of functions $h:[0,T]\times \R_+\to\R_+$ such that $\int_{\R_+} h(t,z)\,dz < \infty$ for every $t \in [0,T]$.

The reinsurer quantifies possible misspecification in $\P_k$, $k\in\N$, with relative entropy penalties. The relative entropy, or Kullback-Leibler (KL) divergence, from $\P_k$ to an absolutely continuous candidate probability measure $\Q^{\boldsymbol{\sigma}}$ is given by
\begin{equation}
    \KL{\Q^{\boldsymbol{\sigma}}}{\P_k} := \E^{\Q^{\boldsymbol{\sigma}}} \left[\log\left(\frac{d\Q^{\boldsymbol{\sigma}}}{d\P_k} \right) \right] \, ,
\end{equation}
where $\E^{\Q^{\boldsymbol{\sigma}}} [\cdot]$ denotes the expectation under $\Q^{\boldsymbol{\sigma}}$. The reinsurer then penalises a candidate probability measure $\Q^{\boldsymbol{\sigma}}$ by the weighted average of the KL divergences from the individual insurers' probability measures to the candidate one, that is,
\begin{equation}
    \sum_{k\in\N} \pi_k \, \KL{\Q^{\boldsymbol{\sigma}}}{\P_k} \, ,
\end{equation}
where $\pi_k$ are weights satisfying $\pi_k \ge 0$ and  $\sum_{k\in\N} \pi_k = 1$. Each weight $\pi_k$ represents the reinsurer's \textit{a priori} belief in insurer-$k$'s model for insurable losses. For example if $\pi_k = 0$, this means that the reinsurer disregards the model used by insurer-$k$. A naive way for the reinsurer to model the losses would be to use a weighted arithmetic mean of the insurers' compensators. We shall see, however, that the optimal behaviour of the reinsurer instead relates to the weighted geometric mean. Thus, for future reference we define the weighted arithmetic mean of the insurers' systemic loss compensators $v^a$ and the weighted geometric mean of the insurers' systemic loss compensators $v^g$ as
\begin{equation*}
v^a(\bz) := \sum_{j\in\N} \pi_j \vsmallskip v_{j}(\bz)\,,
\quad \text{and} \quad
v^g(\bz) := \prod_{j\in\N} v_{j}(\bz)^{\pi_j} \, ,
\end{equation*}
and similarly the weighted means of the idiosyncratic loss compensators as
\begin{equation*}
w^a_{k}(z_k) := \sum_{j\in\N} \pi_j \vsmallskip w_{k,j}(z_k)\,,
\quad \text{and} \quad
w^g_{k}(z) := \prod_{j\in\N} w_{k,j}(z)^{\pi_j} \,, \qquad \text{for } k \in \N  \,.\\
\end{equation*}

The reinsurer maximises their expected wealth, subject to a KL penalty, by choosing safety loadings $\boldeta$ and compensators $\boldsymbol{\sigma}$. As the reinsurer is the first-mover in the Stackelberg game, they are able to anticipate the insurers' responses to their prices.
\begin{optimization}
\label{reinsurersproblem}
Let  insurer-$k$'s demand for reinsurance be parameterised by $\alpha_k^\dagger[\eta_{t,k}]$. Then the reinsurer's problem is
\begin{equation}
    \label{eqn:P-R}
    \sup_{\eta\in\mfC} \inf_{\boldsymbol{\sigma} \in \V} \, \E^{\Q^{\boldsymbol{\sigma}}}  \left[Y_{T} + \frac{1}{\eps} \sum_{k\in\N} \pi_{k} \,\KL{\Q^{\boldsymbol{\sigma}}}{\P_k} \right],
    \tag{$P_R$}
\end{equation}
where $\eps \,>0$ is a parameter representing the reinsurer's overall ambiguity aversion and we recall that the reinsurer's wealth process $Y$ depends on $\eta$.
\end{optimization}

Incorporating ambiguity aversion into the reinsurer's view of the insurance market is different from assessing the robustness of the reinsurer's model, the latter of which could be achieved by stress testing, among other approaches. While stress testing aims at finding worst-case scenarios, ambiguity aversion penalises scenarios according to the reinsurer's point of view. We refer the reader to \textcite{Kroell2024} for stress testing in an insurance context.

The next statement provides the reinsurer's optimal compensators in feedback form for problem \eqref{eqn:P-R}. We show that the optimal safety loadings of the reinsurer are independent of time, which results in the insurers' demand for reinsurance being constant over time. Furthermore, the reinsurer's optimal 
compensators are time independent, and thus the optimal probability measure that the reinsurer uses to determine premia is constant over time. 

We further show that the optimal reinsurance compensators are proportional to the weighted geometric mean of the compensators of the insurers, i.e. $v^g$ and $w_k^g$. Thus, if $\pi_k$ is small, then insurer-$k$'s model only has a small effect on the reinsurer's model, as its weight when calculating both the systemic losses and the idiosyncratic losses is smaller. Furthermore, if $\pi_k = 0$, then the reinsurer's optimal probability measure is independent of insurer-$k$'s model, but still depends on insurer-$k$'s losses through the exponential distortions to both the systemic loss compensator and insurer-$k$'s idiosyncratic loss compensator.
Moreover, all the reinsurer's compensators scale exponentially with the ambiguity $\varepsilon$.

\begin{proposition}
\label{prop:reinsurer_optim}
For $\boldsymbol{c}:=(c_1, \ldots, c_n)\in\R^n$, define the following parameterised set of compensators: 
\begin{align*}
	\varrho_k^\dagger(z_k, c_k) &:=  w_k^g(z_k) \exp \left( \eps \, [z_k - r(z_k, \alpha_k^\dagger[c_k] )]\right) \,, \quad \forall k \in \N \,, \text{ and} \\
	\varsigma^\dagger(\bz,\boldsymbol{c}) 
    :&= 
    v^g(\bz)\; \exp \left\{ \eps \sum_{k\in\N} \left[ z_k - r(z_k,  \alpha_k^\dagger[c_k]) \right]  \right\} \,.
\end{align*}
Further, suppose that the system of non-linear equations in $\boldsymbol{c}$
     \begin{align*}
        1+c_k = &
        \frac{\displaystyle \int_{\R_+} \!\!\! \left(z_k - r(z_k,\alpha_k^\dagger[c_k] )\right) F_k(dz_k) }{ \displaystyle \int_{\R_+} \!\!\!  \partial_c \, r(z_k,\alpha_k^\dagger[c])\big|_{c=c_k}  F_k(dz_k) } \\
        &+ \frac{\displaystyle \int_{\R_+} \!\!\! \partial_c \, r(z_k,\alpha_k^\dagger[c])\big|_{c=c_k} \, \varrho_k^\dagger(z_k, c_k) \, dz_k
        +   \int_{\Rpn} \!\!\! \partial_c \, r(z_k,\alpha_k^\dagger[c])\big|_{c=c_k} \, \varsigma^\dagger(\bz,\boldsymbol{c}) \, d\bz}{ \displaystyle \lambda_k \int_{\R_+} \!\!\!  \partial_c \, r(z_k,\alpha_k^\dagger[c])\big|_{c=c_k} F_k(dz_k) }
        \,, \qquad \forall k\in\N
    \end{align*}
has a solution, which we denote by $\boldsymbol{\eta}^*$.
Then, the reinsurer's optimal compensators in feedback form for problem \eqref{eqn:P-R} are $\bvrho^*(\cdot) = (\varrho_{1}^\dagger(\cdot, \eta^*_1),\ldots,\varrho^\dagger_{n}(\cdot, \eta_n^*))$ and $\varsigma^*(\cdot):=\varsigma^\dagger(\cdot,\boldsymbol{\eta}^*)$. Note that the optimal compensators are independent of time, thus we simply write $\bvrho^* = \bvrho^*_t$, $\varsigma^* = \varsigma^*_t$, and $\boldsymbol{\sigma}^*  = ( \varsigma^*, \bvrho^*)$.
\end{proposition}

\begin{proof}
The reinsurer's value function is 
\begin{align*}
    \Phi(t,y) = \sup_{\eta\in\mfC} \inf_{\boldsymbol{\sigma} \in \mathcal{V}} \, \mathbb{E}^{\Q^{\boldsymbol{\sigma}}}_{t,y}  \Bigg[ Y_{T} + 
    \frac{1}{\eps}  \sum_{j \in\N} \pi_{j} &\Bigg(\, \sum_{k\in\N} \int_t^T \!\! \int_{\R_+} \!\!\! \left(\varrho_{u,k}(z_k) \left[ \log \left( \frac{\varrho_{u,k}(z_k)}{w_{k,j}(z_k)}\right) - 1\right] + w_{k,j} (z_k) \right) dz_k \,du \\
    & \quad + \int_t^T \!\! \int_{\Rpn}  \!\left( \varsigma_u(\bz) \left[ \log \left( \frac{\varsigma_u(\bz)}{v_{j}(\bz)}\right) - 1 \right] + v_{j}(\bz) \right)  d\bz\,  du  \Bigg) \Bigg]\,,
\end{align*}
where $\E^{\Q^{\boldsymbol{\sigma}}}_{t,y}$ denotes the conditional $\Q^{\boldsymbol{\sigma}}$-expectation given that the reinsurer's wealth process satisfies $Y_{t-}=y$. Using the notation of the arithmetic and geometric mean of the compensators $w_k$, $k\in\N$, and $v$, the value function may be written as
\begin{align*}
    \Phi(t,y) = \sup_{\eta\in\mfC} \inf_{\boldsymbol{\sigma} \in \mathcal{V}} \, \mathbb{E}^{\Q^{\boldsymbol{\sigma}}}_{t,y}  \Bigg[ Y_{T} + 
    \frac{1}{\eps} &\Bigg(\,\sum_{k\in\N} \int_t^T \!\! \int_{\R_+} \!\!\! \left(\varrho_{u,k}(z_k) \left[ \log \left( \frac{\varrho_{u,k}(z_k)}{w^g_{k}(z_k)}\right) - 1\right] + w^a_{k} (z_k) \right) dz_k \,du \\
    & \quad + \int_t^T \!\! \int_{\Rpn}  \!\left( \varsigma_u(\bz) \left[ \log \left( \frac{\varsigma_u(\bz)}{v^g(\bz)}\right) - 1 \right] + v^a(\bz) \right)  d\bz\,  du  \Bigg) \Bigg] \,.
\end{align*}

As $\eta_{k,t}$ are deterministic, the optimal $\varsigma_t(\bz)$ and $\varrho_{t,k}(z_k), \, k \in \N$ do not depend on the reinsurer's wealth and are deterministic fields, so we may denote them by $\varsigma(t,\bz)$ and $\varrho_{k}(t,z_k)$. In this case, the conditional expectation simplifies, and we have that
\begin{align*}
    \Phi(t,y) = \sup_{\eta\in\mfC} \inf_{\boldsymbol{\sigma} \in \bar{\mathcal{V}}} \, \Bigg\{ y &+ \int_t^T \sum_{k \in \N} p_k^R\left(\alpha^*_{u,k}, \eta_{u,k}\right) \, du  \\
    & + \frac{1}{\eps} \sum_{k \in \N} \int_t^T \!\! \int_{\R_+} \left( \varrho_{k}(u,z_k) \left[ \log \left( \frac{\ \varrho_{k}(u,z_k)}{w^g_k(z_k)}\right) - 1 -\eps(z_k-r(z_k,\alpha^*_{u,k}))\right] +   w^a_k(z_k) \right) dz_k \, du \\
    &+ \frac{1}{\eps} \int_t^T \!\! \int_{\Rpn} \left( \varsigma(u,\bz) \left[ \log \left( \frac{\varsigma(u,\bz)}{v^g(\bz)}\right) - 1 -\eps\sum_{k\in\N}(z_k - r(z_k,\alpha^*_{u,k}))\right] +   v^a(\bz) \right) d\bz \, du \Bigg\} \, ,
\end{align*}
where $\bar{\mathcal{V}}\subset \mathcal{V}$ is the subset of deterministic admissible fields.
By the dynamic programming principle, the value function $\Phi$ satisfies the Hamilton-Jacobi-Bellman-Isaacs (HJBI) equation:
 \begin{align*}
    0 = \partial_t \Phi(t,y) + \sup_{\boldsymbol{c}\in\R^n_+} \inf_{\substack{\psi \in \G, \\ \xi_1, \ldots, \xi_n \in \mathcal{H}}} \Bigg\{ & \sum_{k\in\N} p_k^R\left(\alpha^\dagger_k[c_k], c_k \right) \partial_y \Phi(t,y) \\
    & +  \sum_{k\in\N} \int_{\R_+}  \left[ \Phi\left(t,y- \big[z_{k}-r(z_{k},\alpha^\dagger_k[c_k])\big] \right) \,  - \Phi(t,y) \right] \xi_k(t,z_k) \, dz_k
    \\
    &+ \int_{\Rpn}  \left[ \Phi\left(t,y-\sum_{k\in\N} \big[z_k-r(z_k,\alpha^\dagger_k[c_k])\big] \right) \,  - \Phi(t,y) \right] \psi(t,\bz) \, d\bz 
    \\
    & + \frac{1}{\eps} \sum_{k\in\N} \int_{\R_+} \left( \xi_{k}(t,z_k) \left[ \log \left( \frac{\xi_{k}(t,z_k)}{w_{k}^g(z_k)}\right) - 1 \right] + w_{k}^a (z_k) \right)  dz_k \\
    & +  \frac{1}{\eps} \int_{\Rpn} \left( \psi(t,\bz) \left[ \log \left( \frac{\psi(t,\bz)}{v^g(\bz)}\right) - 1 \right] + v^a(\bz) \right)  d\bz
    \Bigg\} \, 
\end{align*}
with terminal condition $\Phi(T) =y$.
Since the value function is linear in $y$, we write $\Phi(t,y) = y + \phi(t)$, leading to the ordinary differential equation for $\phi$:
 \begin{align*}
    0 = \phi'(t) + \sup_{\boldsymbol{c}\in\R^n_+} \inf_{\substack{\psi \in \G, \\ \xi_1, \ldots, \xi_n \in \mathcal{H}}} \Bigg\{ & \sum_{k\in\N} p_k^R\left(\alpha^\dagger_k[c_k], c_k \right)
    - \sum_{k\in\N} \int_{\R_+} \!\!  \left[ z_{k}-r(z_{k},\alpha^\dagger_k[c_k]) \right] \xi_k(t,z_k) \, dz_k
    \\
    &- \int_{\Rpn} \sum_{k\in\N} \big[z_k-r(z_k,\alpha^\dagger_k[c_k])\big] \psi(t,\bz) \, d\bz 
    \\
    & + \frac{1}{\eps} \sum_{k\in\N} \int_{\R_+} \left(\xi_{k}(t,z_k) \left[ \log \left( \frac{\xi_{k}(t,z_k)}{w_{k}^g(z_k)}\right) - 1 \right] + w_{k}^a (z_k) \right)  dz_k \\
    & +  \frac{1}{\eps} \int_{\Rpn} \left( \psi(t,\bz) \left[ \log \left( \frac{\psi(t,\bz)}{v^g(\bz)}\right) - 1 \right] + v^a(\bz) \right)  d\bz
    \Bigg\} \, 
\end{align*}
with terminal condition $\phi(T) = 0$.

We first consider the inner optimisation and in a first step optimise over $\psi$, and in a second step over $\xi_1, \ldots, \xi_n$. The infimum problem, where we minimise the functional of $\psi$, is given by
\begin{equation*}
    \mathcal{L}[\psi] := \int_{\Rpn} \!\! \psi(t,\bz) \left[ \log \left( \frac{\psi(t,\bz)}{v^g(\bz)}\right) - 1 - \eps \, \sum_{k\in\N} \big[z_k-r(z_k,\alpha^\dagger_k[c_k])\big] \right] d\bz \, .
\end{equation*}
Let $\delta >0$, and $h$ be an arbitrary function such that $\delta\; h + \psi \in \G$. Applying a variational first order condition, and using $\log\left(\frac{\psi(t,\bz) + \delta\,  h(x)}{v^g(\bz)}\right) = \log\left(\frac{\psi(t,\bz)}{v^g(\bz)}\right) + o( \delta ) $, results in
\begin{align*}
    0 &= \lim_{\delta \to 0} \frac{\mathcal{L}[\psi + \delta \, h] - \mathcal{L}[\psi]  }{\delta} = \int_{\Rpn} h(\bz) \left[ \log \left( \frac{\psi(t,\bz)}{v^g(\bz)}\right) - \eps  \, \sum_{k\in\N} \big[z_k-r(z_k,\alpha^\dagger_k[c_k])\big]   \right] d\bz \, .
\end{align*}
As $h$ was arbitrary, this implies that 
\begin{equation*}
    0 = \left[ \log \left( \frac{\psi(t,\bz)}{v^g(\bz)}\right) - \eps  \, \sum_{k\in\N} \big[z_k-r(z_k,\alpha^\dagger_k[c_k])\big]   \right]\,, \, \quad \text{for all } \bz \in \Rpn \text{ and } t\in[0,T].
\end{equation*}
Solving the above equation for $\psi$ gives the optimal compensator in feedback form, i.e. $\vs^*$:
\begin{equation*}
    \varsigma^*_t(\bz, \boldsymbol{c} ) =  v^g(\bz) \, \exp \left\{ \eps \, \sum_{k\in\N} \big[z_k-r(z_k,\alpha^\dagger_k[c_k])\big] \right\} \, .
\end{equation*}

In the second step, we use a similar procedure to solve individually for each $\xi_k$, $k\in\N$, giving each optimal compensator in feedback form, i.e. $\varrho^*_k$:
\begin{equation*}
\label{eq:reinsurer_prob}
    \varrho^*_{t,k}(z_k,c_k) = w_k^g(z_k) \, \exp \left\{ \eps \, \big[z_k-r(z_{k},\alpha^\dagger_k[c_k])\big] \right\} \, .
\end{equation*}
Substituting this into the HBJI equation, we have that
\begin{equation*}
        0 = \phi'(t) + \sup_{\boldsymbol{c}\in\R^n_+} \Bigg\{  \sum_{k\in\N} p_k^R\left(\alpha^\dagger_{k}[c_k], c_k\right) 
        + \frac{1}{\eps} \sum_{k\in\N} \int_{\R_+}   \! \! \left[w_k^a(z_k) - \varrho^*_{t,k}(z_k, c_k )  \right]  dz_k
        + \frac{1}{\eps} \int_{\Rpn}  \! \left[v^a(\bz) - \varsigma^*_t(\bz, \boldsymbol{c} )  \right]  d\bz \Bigg\} \, .
\end{equation*}
The derivative of the term in curly braces with respect to $c_k$, for each $k\in\N$, is, recalling the expression for the reinsurance premium given in \eqref{eqn:reins_prem}, 
\begin{align*}
    \frac{\partial}{\partial c_k} \{ \cdots \} =& \lambda_k \left[
     \int_{\R_+} \!\!\! \left[ z_k - r(z_k,\alpha_k^\dagger[c_k] )\right] F_k(dz_k)
    - (1+c_k) \! \int_{\R_+} \!\!\! \partial_{c_k} r(z_k,\alpha_k^\dagger[c_k])F_k(dz_k) \right]
    \\
    & + \int_{\R_+}  \partial_{c_k}r(z_k,\alpha_k^\dagger[c_k]) \, \varrho_{t,k}^*(z_k, c_k) \,dz_k
    + \int_{\Rpn}
    \partial_{c_k}r(z_k,\alpha_k^\dagger[c_k]) \,\varsigma^*_t(\bz, \boldsymbol{c}) \,d\bz\,.
\end{align*}
The first order conditions require setting this expression to $0$ for all $k\in\N$ and all $t\in[0,T]$. Hence, we obtain that the optimal $\eta_{k,t}^*$ satisfies the non-linear equation
\begin{align*}
    \eta_{k,t}^* &= 
    \frac{\displaystyle\int_{\R_+} \!\!\! \left(z_k-r(z_k,\alpha_k^\dagger[\eta_{k,t}^*])\right) F_k(dz_k)}{ \displaystyle\int_{\R_+}  \!\!  \partial_c \, r(z_k,\alpha_k^\dagger[c])\big|_{c=\eta_{k,t}^*} \,F_k(dz_k)} \\
    & \quad + \frac{\displaystyle \int_{\R_+} \!\!  \partial_a \, r(z_{k},a)\big|_{a=\alpha_k^\dagger[\eta_{k,t}^*]} \, \varrho^*_{t,k}(z_k,\eta_{k,t}^*) dz_{k}  + \int_{\Rpn} \partial_a \, r(z_k,a)\big|_{a=\alpha_k^\dagger[\eta_{k,t}^*]} \,  \varsigma^*_t(\bz,\boldeta^*) \, d\bz}{\lambda_k \displaystyle\int_{\R_+}  \!\!  \partial_a \, r(z_k,a)\big|_{a=\alpha^\dagger[\eta_{k,t}^*]} \,F_k(dz_k)} -1   \,. 
\end{align*}
As there is no explicit time dependence in the above equation, the roots of these non-linear equations are constant over time. 
\end{proof}

The following theorem gives the Stackelberg equilibrium for the insurers-reinsurer game. The proof follows by combining \cref{prop:insurer_optim}, which gives the optimal behaviour of the insurers, and \cref{prop:reinsurer_optim}, which gives the optimal behaviour of the reinsurer.

\begin{theorem}
\label{thm:stackelberg_eq}
Under the assumptions of \cref{prop:insurer_optim} and \cref{prop:reinsurer_optim}, the Stackelberg equilibrium is $(\boldalpha^*, \boldeta^*, \varsigma^*, \boldsymbol{\varrho}^*)$, where $\boldalpha^*:= (\alpha_1^*, \ldots, \alpha_n^*)$, with $\alpha_k^*:= \alpha^\dagger_k[\eta_k^*]$ for all $k\in\N$, $\boldsymbol{\varrho}^* = (\varrho_1^*(\cdot), \ldots, \varrho_n^*(\cdot))$ where for each $k \in \N$
\begin{equation*}
    \varrho^*_{k}(z_k) = w_k^g(z_k) \, \exp \left\{ \eps \, \big[z_k-r(z_{k},\alpha^*_k)\big] \right\} \, ,
\end{equation*}
and
\begin{equation*}
    \varsigma^*(\bz) = v^g(\bz) \, \exp \left\{ \eps \, \sum_{k\in\N} \big[z_k-r(z_k,\alpha^*_k)\big] \right\} \, .
\end{equation*}
The $\boldeta^* = (\eta_1^*, \ldots, \eta_n^*)$ satisfy the system of equations
\begin{align*}
        0 &=\int_{\R_+}  \!\! \partial_a \, r(z_k,a)\big|_{a=\alpha^\dagger[\eta_k^*]} \left\{(1+\eta_k^*)-e^{\gamma_k r(z_k,\alpha^\dagger[\eta_k^*])}\right\} F_k(dz_k)  \,, \quad \forall k\in\N \, ,  \\
        1+\eta_{k}^* &= 
    \frac{\displaystyle\int_{\R_+} \!\!\! \left(z_k-r(z_k,\alpha_k^\dagger[\eta_{k}^*])\right) F_k(dz_k) }{ \displaystyle\int_{\R_+} \!\!  \partial_c \, r(z_k,\alpha_k^\dagger[c])\big|_{c=\eta_{k}^*} \,F_k(dz_k) } \\
    & \quad + \frac{\displaystyle \int_{\R_+} \!\!  \partial_a \, r(z_{k},a)\big|_{a=\alpha_k^\dagger[\eta_k^*]} \, \varrho^*_{k}(z_k) \, dz_{k}  +  \int_{\Rpn} \partial_a \, r(z_k,a)\big|_{a=\alpha_k^\dagger[\eta_k^*]} \,  \varsigma^*(\bz) \, d\bz}{ \displaystyle \lambda_k \int_{\R_+} \!\!  \partial_a \, r(z_k,a)\big|_{a=\alpha^\dagger[\eta_k^*]} \,F_k(dz_k) }
    \,, \quad \forall k\in\N   \,. 
\end{align*}
\end{theorem}

One observes that the reinsurer's optimal compensators are exponential distortions of the weighted geometric mean of the insurers' compensators. In particular, the optimal compensator for insurer-$k$'s idiosyncratic loss process is the weighted geometric mean of the compensators for that loss process under all models, distorted by the losses insurer-$k$ cedes to the reinsurer. The optimal compensator for the systemic loss is the weighted geometric mean of the compensators of the systemic loss under all models, distorted by the sum of all losses ceded to the reinsurer.

For small ambiguity-aversion of the reinsurer, $\eps \to 0$, each compensator approaches the weighted geometric mean. We see that the Stackelberg equilibrium $(\boldalpha^*, \boldeta^*, \varsigma^*,\boldsymbol{\varrho}^*)$ is independent of time, as the premia, the reinsurance demand, and the reinsurance optimal probability measure are all constant over time. 
Note that in this general case, the optimal insurance demand $\boldalpha^*$ is a function of the optimal reinsurance safety loadings $\boldeta^*$ and the optimal reinsurance compensators $\varsigma^*$, $\bvrho^*$. Thus, the amount of reinsurance each insurer buys depends on the reinsurer's ambiguity aversion $\eps$ and the reinsurer's belief in each of the insurer's loss models, i.e. on $\pi_k$, $k \in \N$. 

The connection between $\boldalpha^*$, $\boldeta^*$, $\varsigma^*$, and $\bvrho^*$ is difficult to disentangle for general reinsurance contracts. However, for concrete reinsurance contracts, e.g., proportional and excess-of-loss, and a market with two insurers, we can illustrate the Stackelberg equilibrium, as we do in the next section.

We conclude this section with a discussion on the differences between our approach and an alternate approach where the reinsurer solves separate Stackelberg games with each insurer. 

\begin{remark}[Comparison with separate games]
\label{remark:comparison}
    It is natural to compare our equilibrium (the single game) with one where the reinsurer plays $n$ separate Stackelberg games, one with each insurer. Note that the equilibrium for the game with only a single insurer is given by taking $n=1$ in \cref{thm:stackelberg_eq}.

    Comparing the two approaches, we see that in the single game setting, the reinsurer uses all available information to create a single model for the insurable losses, and uses that model to price reinsurance contracts. For example, the model for the systemic loss, $\vs^*(z)$, depends on all $n$ models for the systemic loss, as well as all possible losses stemming from all reinsurance contracts. This model is then used to determine all reinsurance prices, as given by the safety loading, $\eta^*_k$, $k \in \N$.

    In contrast, if the reinsurer solves $n$ separate games, they only use a single input model $\P_k$ when pricing the losses for insurer-$k$. The resulting optimal pricing model for insurer-1 is $\vs_1^*(z) = v_1(z) \exp \left(\eps \big[z_1-r(z_1,\alpha^*_1)\big] \right)$, and similar for the other insurers.  In other words, the reinsurer determines the price for insurer-1 ignoring the information about the loss distribution they may have from insurers-$2, \ldots, n$. The resulting optimal pricing model is then different for each insurer, so the reinsurer is using different models to price each contract. Furthermore, as the reinsurer prices each contract separately, they ignore the fact that they are selling reinsurance on some losses to multiple customers, which is captured by the term $\sum_{k\in\N} \big[z_k-r(z_k,\alpha^*_k)\big]$ in the optimal model in the single game. Therefore in the single game setting, the reinsurer is not appropriately accounting for the risk in the case where the losses are highly correlated and on the other hand, is ignoring the possible benefits of diversification.

    In the single game approach, the reinsurer is required to develop a joint model for the systemic losses, which could entail more model risk. However, this is addressed by the fact that the single game approach incorporates multiple models, which mitigates the effect of inaccuracies in any one model.
\end{remark}

Thus, while the reinsurer could choose to play separate Stackelberg games with each insurer, in doing so they would ignore relevant information about insurable losses when creating their pricing model. They would also be unaware of possible gains they could make from diversification or possible losses they could be exposed to through risk concentration.

\section{Application to reinsurance contract types}
\label{sec:application}

In this section, we calculate and illustrate the Stackelberg equilibrium for specific reinsurance contracts. In particular, we consider proportional reinsurance and excess-of-loss reinsurance, with and without a policy limit. We further illustrate the solutions numerically in the case of two insurers.

\subsection{Proportional reinsurance}
First, we explicitly calculate the optimal proportional reinsurance contract. Recall that for this type of contract, the insurer decides a proportion of the loss to retain, $\alpha_k \in (0,1]$, while the reinsurer covers the remaining proportion, $(1-\alpha_k)$.

\begin{proposition}[Proportional reinsurance]
    Let $\boldalpha^*$ denote the solution of the following system of equations for $k \in \N$:
    \begin{equation*}
        0 = \,\int_{\R_+} \!\! z_k \, e^{\gamma_k \alpha^*_k z}\, \left[(1-\alpha^*_k)\, \gamma_k \, z_k - 1 \right] \lambda_k F_k(dz_k) 
    + \int_{\R_+} \!\! z_k \, \varrho_k^*(z_k)  \, dz_k
    + \int_{\Rpn} \!\! z_k \, \varsigma^*(\bz)  \, d\bz\,,
    \end{equation*}
    where $\varrho_k^*$, $k \in \N$, and $\varsigma^*$ are given in \Cref{thm:stackelberg_eq} with $z-r(z,\alpha) = (1-\alpha) z$.

    If such $\boldalpha^*$ exist, then the equilibrium safety loadings are
    \begin{equation*}
        \eta_k^* = \frac{\int_{\R_+} \!\! z_k \, e^{\gamma_k \alpha^*_k z}\, F_k(dz_k) }{\int_{\R_+} \!\! z_k \, F_k(dz_k) } - 1\,, \quad \forall k\in\N \, .
    \end{equation*}
    The Stackelberg equilibrium for proportional reinsurance is $(\boldalpha^*$, $\boldeta^*, \varsigma^*, \bvrho^*)$.
\end{proposition}

To further analyse the Stackelberg equilibrium in the proportional reinsurance case, we assume that there are two insurers whose losses are Gamma distributed with different parameters. We will first examine the simplified setting where the systemic losses are comonotonic and there are no idiosyncratic losses, before considering a more complex example. The assumption of comonotonic systemic losses means that the model for the systemic loss is a function of a single scalar variable.

\begin{example}
\label{ex:prop_Gamma}
    Suppose there are two insurers and the systemic losses are comonotonic and identically marginally Gamma distributed under both insurers' models. However, the insurers disagree on the model parameters. The marginal losses are Gamma distributed with shape parameters $m_1,m_2>0$, scale parameters $\xi_1, \xi_2>0$, and intensities $\lambda_1, \lambda_2 > 0$, under insurer-1 and insurer-2's models, respectively.
    Let $\xi_k \leq \min \{\frac{1}{\gamma_k},\,\frac{1}{2\eps} \}$ for $k=1,2$ and define
    \begin{align*}
        \tilde m &:=  m_1\pi_1  + \pi_2 m_2 \, , \\
        \tilde \xi &:= \frac{\xi_1 \xi_2}{\pi_1 \xi_2 + \pi_2 \xi_1 - \eps \, \xi_1 \xi_2 (2-\alpha^*_1 - \alpha^*_2)} \, .
    \end{align*}
    Then, at the Stackelberg equilibrium, the reinsurer's severity distribution is a Gamma distribution with shape $\tilde m$ and scale $\tilde \xi$. This means that the reinsurer calculates premia as if both insurers have a Gamma distribution with parameters $\tilde m$ and $\tilde \xi$ and constant intensity 
    \begin{equation*}
        \Lambda :=   \tilde \xi^{\tilde m} \, \Gamma(\tilde m)  \left( \frac{\lambda_1}{\Gamma(m_1) \xi_1^{m_1}} \right)^{\pi_1} \!\! \left( \frac{\lambda_2}{\Gamma(m_2) \xi_2^{m_2}} \right)^{\pi_2} \, ,
    \end{equation*}
    where $\Gamma(\cdot)$ is the Gamma function. The assumption $\xi_k \leq \min \{\frac{1}{\gamma_k},\,\frac{1}{2\eps} \}$ implies that the insurer only underwrites losses that are in line with their absolute risk aversion. Recall that the mean of a Gamma random variable is $m_k \xi_k$, thus insurer-$k$ only considers risks with mean $m_k \xi_k \le \frac{m_k}{\gamma_k}$. Similarly, the reinsurer requires that losses be in line with their ambiguity aversion, for example, they require that the means of each of the insurers' losses be bounded proportionally to the reciprocal of their ambiguity, i.e. $m_1\xi_1 + m_2 \xi_2 \le \frac{m_1 + m_2}{2 \eps}$.

    The insurers choose to reinsure proportions $\alpha_k^*$ of their losses, where $\alpha^*_k$, $k=1,2$, is the solution of the following system of equations:
     \begin{equation*}
         0 = \frac{-1}{(1 - \alpha^*_k \,\gamma_k \, \xi_k)^{m_k + 1}} 
         +   \frac{\gamma_k (1+m_k) \, \xi_k \, (1-\alpha^*_k)}{(1-\alpha_k^* \, \gamma_k\, \xi_k)^{m_k + 2}} 
         +  \frac{\tilde m\,\tilde\xi \, \Lambda}{m_k\,\xi_k \, \lambda_k} \, , \quad k = 1,2 \, .
     \end{equation*}
    The reinsurer charges safety loadings
    \begin{equation*}
        \eta^*_k = \frac{1}{(1-\alpha^*_k \vsmallskip \gamma_k \, \xi_k)^{m_k + 1}} -1 \, , \quad k = 1,2 \, .
    \end{equation*}

  In the numerical example, the insurers'  shape and scale parameters and intensities are $m_1 = 1.5, \, \xi_1 = 1, \, \lambda_1 = 2$ and $m_2 = 2, \, \xi_2 =1.25, \, \lambda_2 = 2.5$, respectively.   In particular, insurer-2 has greater expected losses than insurer-1, as it has on average more frequent claims, and the mean of their claim distribution is larger. The insurers have the same risk aversion, i.e. $\gamma_1=\gamma_2=0.5$. 

  \begin{figure}[hbt!]
    \centering
    \subfloat[Reinsurer's severity distribution for different $\varepsilon$ and $\pi_1 = \pi_2 = 0.5$]{\scalebox{0.85}{\includegraphics{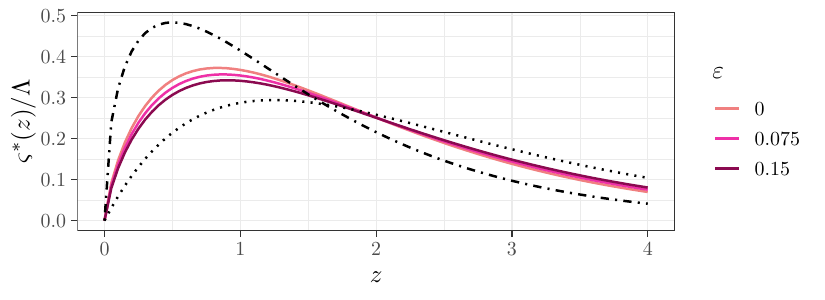}} \label{fig:1d_comp_prop}}
    \hfill
    \subfloat[Reinsurer's severity distribution for different $\pi_2$ and $\eps = 0.1$]{\scalebox{0.85}{\includegraphics{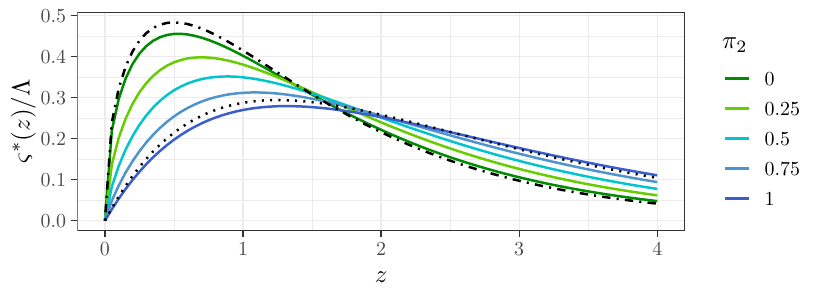}}}
    \caption{The severity distributions of the systemic loss under of insurer-1's model (dot-dashed line), insurer-2's model (dotted line), and the reinsurer's model (solid lines). The insurers have risk aversion parameters $\gamma_1=\gamma_2=0.5$. The systemic losses are comonotonic under both models, with identical Gamma distributed marginals. Under models 1 and 2, respectively, the intensities are $\lambda_1 = 2,  \, \lambda_2 = 2.5$, the scale parameters are $\xi_1 = 1, \, \xi_2 =1.25$ and the shape parameters are $m_1 = 1.5, \, m_2 = 2$. In panel (a), $\eps$ is varied while the weights are fixed at $\pi_1 = \pi_2 =0.5$. In panel (b), $\eps = 0.1$ and the weights are varied.}
    \label{fig:compensator_Gamma}
  \end{figure}

  \begin{figure}[hbt!]
    \centering
    \subfloat[Reinsurer's total intensity for different $\varepsilon$]{\scalebox{0.9}{\includegraphics{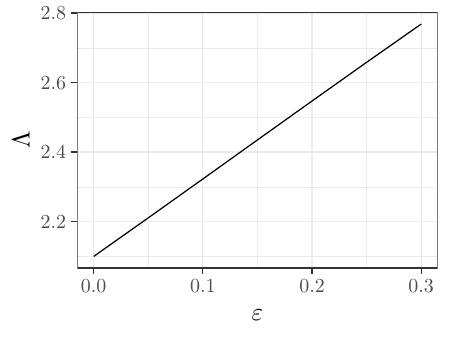}}}
    \hspace{1em}
    \subfloat[Reinsurer's total intensity for different $\pi_2$]{\scalebox{0.9}{\includegraphics{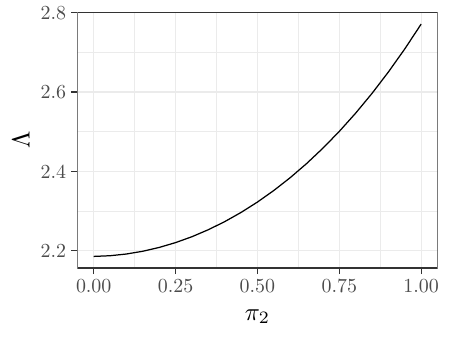}}}
    \caption{The total intensity, $\Lambda$, of the reinsurer.  In panel (a), the reinsurer's ambiguity aversion parameter, $\eps$, is varied while the weights are fixed at $\pi_1 = \pi_2 =0.5$. In panel (b), $\eps = 0.1$ and the weights applied to each insurer's loss model are varied. The insurers have risk aversion parameters $\gamma_1=\gamma_2=0.5$. The systemic losses are comonotonic under both models, with identical Gamma distributed marginals. Under models 1 and 2, respectively, the intensities are $\lambda_1 = 2,  \, \lambda_2 = 2.5$, the scale parameters are $\xi_1 = 1, \, \xi_2 =1.25$, and the shape parameters are $m_1 = 1.5, \, m_2 = 2$.}
    \label{fig:Lambda_Gamma}
  \end{figure}

  \cref{fig:compensator_Gamma} plots the optimal severity distributions (normalised compensators) under insurer-1's model (dot-dashed line), insurer-2's model (dotted line), and the reinsurer's model (solid lines). Panel (a) shows the effect of the ambiguity aversion $\eps$ on the severity distribution. Here, the weights applied to the insurers' loss models are kept fixed at $\pi_1 = \pi_2 = 0.5$. The salmon-coloured line gives the reinsurer's severity distribution when the reinsurer is not ambiguity-averse ($\eps=0$). In this case, the reinsurer's severity distribution is the normalised geometric mean of the insurers' distributions. As the reinsurer's ambiguity aversion increases, they place more weight on larger losses. This can be seen in the figure in the pink and purple lines ($\eps = 0.075, \, 0.15$). Panel (b) shows the effect of the choice of the weights, $\pi_1$ and $\pi_2$, on the severity distribution, when $\eps$ is kept fixed at 0.1. As the weights vary, the reinsurer interpolates between the two insurers' severity distributions; however, they give a larger weight to large losses $z$ due to their ambiguity aversion.

  \cref{fig:Lambda_Gamma} shows the reinsurer's total constant intensity, $\Lambda$, as a function of the reinsurer's ambiguity aversion, $\eps$, (panel (a)) and the weights applied to insurer-2's loss model, $\pi_2$ (panel (b)). Panel (a) shows that as the reinsurer becomes more ambiguity averse (increasing $\eps$), the rate of arrival of claims increases under their optimal model. Panel (b) shows that the total intensity is also increasing as the reinsurer puts more weight on the model with greater expected losses.

    \begin{figure}[b!]
    \centering
    \subfloat[Reinsurer's severity distribution, separate pricing approach]{\scalebox{0.85}{\includegraphics{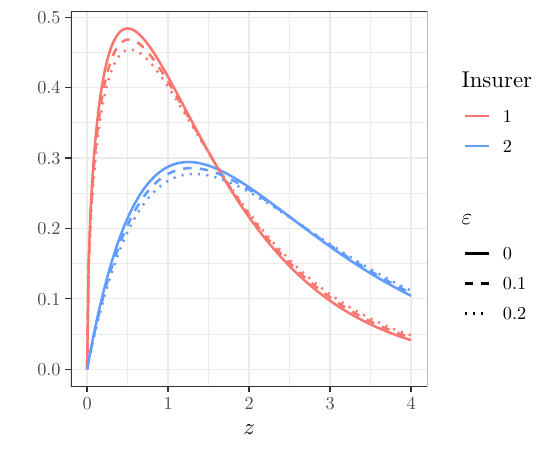}} \label{fig:seperate_games_a}}
    \hfill
    \subfloat[Reinsurer's severity distribution, shared pricing approach]{\scalebox{0.85}{\includegraphics{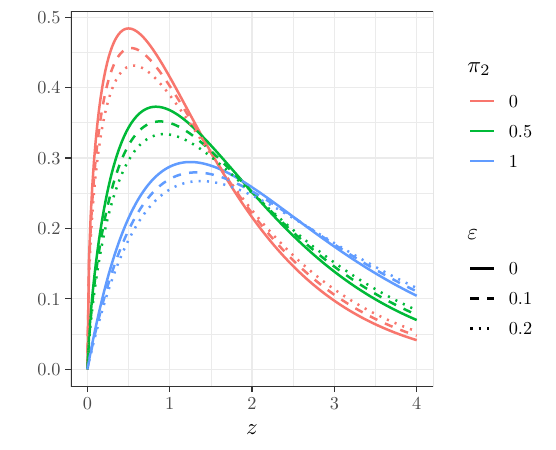}}}
    \caption{The severity distributions of the systemic loss for different values of the reinsurer's ambiguity aversion, $\eps$, under the assumption of (a) separate Stackelberg games played between the reinsurer and each insurer and (b) a single Stackelberg game played between the reinsurer and both insurers. The insurers have risk aversion parameters $\gamma_1=\gamma_2=0.5$. The systemic losses are comonotonic under both models, with identical Gamma distributed marginals. Under models 1 and 2, respectively, the intensities are $\lambda_1 = 2,  \, \lambda_2 = 2.5$, the scale parameters are $\xi_1 = 1, \, \xi_2 =1.25$, and the shape parameters are $m_1 = 1.5, \, m_2 = 2$.}
    \label{fig:separate_games}
  \end{figure}

  We can use this simplified setting to examine the effect of the reinsurer treating the reinsurance pricing with ambiguity problem as a single Stackelberg game, as opposed to two separate games, as discussed in \cref{remark:comparison}. Assuming the reinsurer instead solves separate games with insurer-1 and insurer-2, they then use the following models: for insurer-1, a Gamma distribution with shape parameter $m_1$ and scale parameter $\xi_1 / (1- \eps_1 \xi_1 (1-\alpha_1^*))$ and intensity $\lambda_1/(1- \eps_1 \xi_1 (1-\alpha_1^*))^{m_1}$; for insurer-2, a Gamma distribution with shape parameter $m_2$ and scale parameter $\xi_2 / (1- \eps_2 \xi_2 (1-\alpha_2^*))$ and intensity $\lambda_2/(1- \eps_2 \xi_2 (1-\alpha_2^*))^{m_2}$. 

  Thus when pricing reinsurance for insurer-1, the reinsurer's model for insurer-1 depends only on insurer-1's model and the reinsurance demand of insurer-1. The reinsurer completely disregards insurer-2's model when pricing insurer-1's contract, which has the same effect as placing no weight on it ($\pi_2=0$). Furthermore, the reinsurer's model also does not account for the reinsurance it sells to insurer-2 (effectively setting $\alpha_2^*=0$). A similar statement holds for insurer-2.

  \cref{fig:separate_games} illustrates these effects.  Panel (a) shows the model used for the systemic losses when the reinsurer uses the separate pricing approach, i.e., plays separate games with each insurer, while panel (b) shows the model used for the systemic losses when the reinsurer uses the shared pricing approach, i.e., plays a single game with both insurers, for different weights placed on the two distributions. We observe that in the shared pricing approach, the reinsurer can choose how to weigh the two input models, varying from only using insurer-1's model (red lines), only using insurer-2's model (blue lines), to using a combination of both models (green lines). In fact, there is a continuum of possible models, depending on the choice of the weight $\pi_2$. In contrast, in the separate pricing approach, the reinsurer must ignore the other model when pricing one insurer.

  Furthermore, even when the reinsurer disregards one of the two models in the single game ($\pi_2=0$ to ignore model 2 or $\pi_2=1$ to ignore model 1), there is still a difference compared to when they use the separate pricing approach, as long as the reinsurer is ambiguity averse ($\eps >0$). Examining the dashed and dotted lines ($\eps=0.1, 0.2$) in both panels for the corresponding red and blue curves, we see that the effect of the ambiguity aversion is more pronounced in the shared pricing approach (panel (b)). This is because the reinsurer accounts for the fact that two reinsurance contracts have been written on the loss when playing the single game, rather than just one. 
  
  This fact is further emphasised in \cref{tab:eta-compare}, which compares the optimal safety loadings $\eta_k$, $k=1,2$, in the separate pricing approach and the shared pricing approach, when the reinsurer ignores the other insurer's model. We observe that the safety loadings are the same only when there is no ambiguity ($\eps=0$), as expected. When there is ambiguity ($\eps>0$), the safety loadings charged in the shared pricing approach are greater, accounting for the fact that two reinsurance contracts have been written. In sum, by setting all reinsurance prices at once, the reinsurer uses all available information about the loss distributions and total reinsurance demand, allowing them to more appropriately price the contracts.

  \begin{table}[ht]
    \centering
    \begin{tabular}{l@{\hskip 0.4in} c c@{\hskip 0.3in} cc}
    & \multicolumn{2}{p{0.2\textwidth}}{\centering $\eta_1^*$} & \multicolumn{2}{p{0.2\textwidth}}{\centering$\eta_2^*$} \\
    \toprule
    & separate & shared ($\pi_2=0$) & separate & shared ($\pi_2=1$) \\
    \midrule
    $\eps = 0$ & 1.726 & 1.726 & 5.926 &  5.926 \\ 
    $\eps = 0.1$ & 1.802 & 1.866 &6.033 & 6.117 \\ 
    $\eps = 0.2$ & 1.878 & 2.012 & 6.148 & 6.305 \\
    \bottomrule
    \end{tabular}
    \caption{Comparison of optimal safety loadings $\eta_1^*$ and $\eta^*_2$ under the separate and shared pricing assumptions. The insurers have risk aversion parameters $\gamma_1=\gamma_2=0.5$. The systemic losses are comonotonic under both models, with identical Gamma distributed marginals. Under models 1 and 2, respectively, the intensities are $\lambda_1 = 2,  \, \lambda_2 = 2.5$, the scale parameters are $\xi_1 = 1, \, \xi_2 =1.25$, and the shape parameters are $m_1 = 1.5, \, m_2 = 2$.}
    \label{tab:eta-compare}
    \end{table}

\end{example}

We note that one can extend this example to $n$ insurers, as long as $\xi_k \leq \min \{\frac{1}{\gamma_k},\,\frac{1}{n\eps} \}$ for $k\in\N$. In this case one solves a similar $n$-dimensional system of equations for $\boldalpha$, the reinsurer's severity distribution remains Gamma distributed, and the equations to obtain $\eta_k^*$ from $\alpha_k^*$ remains the same. The constraints on the scale parameters $\xi_k$ mean that the insurers' loss models must be sufficiently light-tailed that they offer insurance according to their own utility preferences as parameterised by $\gamma_k$, and sufficiently light-tailed that the reinsurer offers reinsurance according to their ambiguity aversion, $\eps$.

Next we consider a more complex model, where there are idiosyncratic losses in addition to the systemic losses. Furthermore, the copula of the systemic losses is a Gumbel copula under both models.

\begin{example}
\label{ex:prop_Gumbel}
Suppose there are two insurers and that, under both insurers' models, the systemic losses are marginally identically Gamma distributed with shape parameter $m^S$ and scale parameter $\xi^S$ and with a dependence structure given by a Gumbel copula with parameter $\theta = 10/6$. This means the systemic losses of insurer-1 and insurer-2 are upper-tail dependent and have a Kendall's $\tau$ coefficient of 0.4 under both models. We further assume that under both models each insurer's idiosyncratic losses are identically Gamma distributed with shape parameter $m^I$ and scale parameter $\xi^I$. The parameters used in the numerical illustration are given in \cref{tab:prop_Gumbel_params}. We note that insurer-2's model has greater expected losses.

    \begin{table}[b!]
        \centering
        \begin{tabular}{l @{\hskip 0.4in} l l l l l l l}
            Parameter & $\theta$ & $\lambda^S$ & $m^S$ & $\xi^S$ & $\lambda^I$ & $m^I$ & $\xi^I$\\
            \toprule
            Insurer-1's model & 1.67 & 2 & 1.5 & 1 & 1.67 & 1.25 & 1\\
            Insurer-2's model & 1.67 & 2.5 & 2 & 1.25 & 2 & 1.5 & 1\\
        \end{tabular}
        \caption{Parameters for the numerical illustration in \cref{ex:prop_Gumbel}. The systemic losses have a Gumbel copula with parameter $\theta$ and are marginally Gamma distributed with shape parameter $m^S$ and scale parameter $\xi^S$. The idiosyncratic losses are Gamma distributed with shape parameter $m^I$ and scale parameter $\xi^I$.}
        \label{tab:prop_Gumbel_params}
    \end{table}
    
    \cref{fig:prop_2ins} shows the proportion of each loss retained by each insurer at equilibrium, $\alpha_k^*$, and the safety loading charged to them, $\eta_k^*$, as a function of $\pi_2$, the weight applied by the reinsurer to insurer-2's loss model. Three different values of the ambiguity aversion parameter, $\eps=0,\,0.05, \, 0.1$, are shown. Consider first the solid line, $\eps=0$, where the reinsurer has no ambiguity aversion. We observe that as $\pi_2$ increases, meaning that the reinsurer puts more weight onto the model with greater expected losses, the safety loadings, $\eta^*_k$, for both insurers are increasing. The amount retained by each insurer, $\alpha^*_k$, increases as well. 

    \begin{figure}[hbt!]
        \centering
        \scalebox{0.97}{\includegraphics{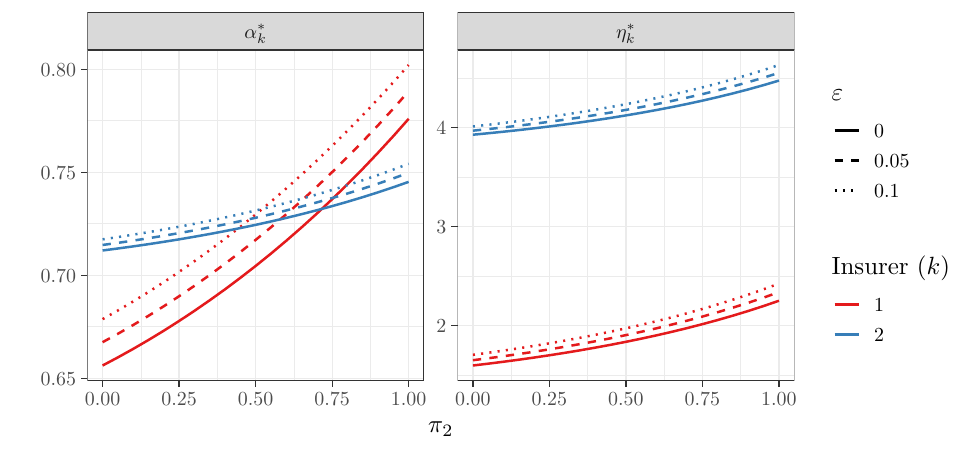}}
        \caption{The equilibrium proportion of losses each insurer retains, $\alpha_k^*$, and the equilibrium safety loadings charged to insurer-$k$, $\eta_k^*$, $k=1,2$, as a function of the weight the reinsurer places on insurer-2's loss model, $\pi_2$. The insurers have risk aversion parameters $\gamma_1=\gamma_2=0.5$ and their severity distributions are Gamma distributed with parameters given in \cref{tab:prop_Gumbel_params}.}
        \label{fig:prop_2ins}
    \end{figure}
    
    In particular, we see that when both loss models are used to calculate the insurance premium ($\pi_2 \in (0,1)$), insurer-1 (red lines) pays more than when only their loss model is used ($\pi_2=0$). Meanwhile, insurer-2 (blue lines) benefits from the use of both models, paying less when $\pi_2 \in (0,1)$ than when $\pi_2 =1$. We further observe that insurer-1 is more sensitive to the pricing model than insurer-2. Insurer-1 buys reinsurance for 34\% of the loss when $\pi_2 = 0$ and reduces their reinsurance demand to 22\% for $\pi_2 = 1$, whereas insurer-2 only decreases their demand for reinsurance from 29\% to 25\%. When the reinsurer is ambiguity-averse (the dashed lines, $\eps = 0.05, \, 0.1$), they increase the safety loadings for both insurers and consequently, both insurers buy less reinsurance, i.e., they retain more, as $\alpha_k^*$ is increasing with $\eps$.

    \cref{fig:prop_2dcomp} shows a contour plot of the systemic bivariate loss density. In the contours shown in the bottom left of the figure, which correspond to a density value of 0.18, the $\eps=0.2$ curve (blue) is contained within the $\eps=0.01$ curve (red). The pattern reverses as we move towards the upper-right corner of the figure. For the furthest contours, which correspond to a density value of 0.03, the $\eps=0.01$ curve (red) is now contained within the $\eps=0.2$ curve (blue). Therefore, as $\eps$ increases, the reinsurer puts more weight on higher losses. This is similar to the observation in the previous example, where the systemic losses were comonotonic (\cref{fig:1d_comp_prop}).

    \begin{figure}[hbt!]
        \centering
        \includegraphics{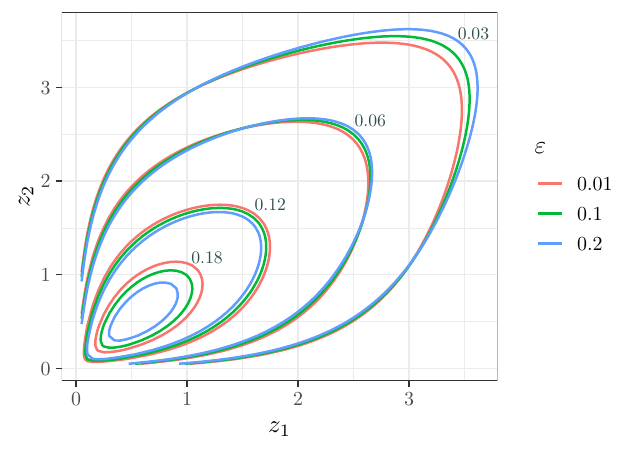}
        \caption{A contour plot of the density of the systemic loss distribution for three different values of the reinsurer's ambiguity aversion parameter, $\eps$. The insurers have risk aversion parameters $\gamma_1=\gamma_2=0.5$ and their severity distributions are Gamma distributed with parameters given in \cref{tab:prop_Gumbel_params}.}
        \label{fig:prop_2dcomp}
    \end{figure}
\end{example}

\subsection{Excess-of-loss reinsurance}
Next, we determine the optimal excess-of-loss reinsurance contract. Recall that for this type of contract, each insurer chooses a retention level, $\alpha_k^*$, beyond which the reinsurer covers the losses.

\begin{proposition}[Excess-of-loss reinsurance]
\label{prop:XL}
    Define insurer-$k$'s survival function, $\overline{F}_k(z_k) := 1 - F_k(z_k)$. Let $\boldalpha^*$ denote the solution of the following system of equations for $k \in \N$:
    \begin{equation*}
        \alpha_k^* = \frac{1}{\gamma_k}
        \log\left(\frac{\displaystyle \int_{\alpha_k^*}^\infty \varrho^*_k(z_k) \, dz_k + \int_{\Rpn} \!\! \Id_{\{z_k \geq \alpha_k^* \}} \, \vs^*(\bz) \, d\bz }{\displaystyle \lambda_k \left[\overline{F}_k(\alpha_k^*) - \gamma_k \int_{\alpha_k^*}^\infty \overline{F}_k(z_k) \, dz_k \right]} \right) \,,
    \end{equation*}
    where $\varrho_k^*$, $k \in \N$, and $\varsigma^*$ are given in \Cref{thm:stackelberg_eq} with $z-r(z,\alpha) = (z-\alpha)_+$.
    If such $\boldalpha^*$ exist, then 
    \begin{align*}
        \eta_k^* &= e^{\gamma_k \alpha^*_k} - 1 \,, \qquad \forall k\in\N  \,,
    \end{align*}
    and the solution to the Stackelberg game is $(\boldalpha^*,\boldeta^*,\varsigma^*, \bvrho^*)$.
\end{proposition}

Recall that, for fixed ambiguity aversion $\eps$, and fixed weights $\pi_k,\,k \in \N$, the reinsurance demand  $\alpha_k^* = \alpha_k^\dagger[\eta_k^*]$, is a function of the optimal safety loading $\eta_k^*$. Interestingly, the safety loading of the premia for insurer-$k$ depends on insurer-$k$'s absolute risk aversion $\gamma_k$. Thus, an insurer with larger risk aversion is willing to pay a larger markup for the same contract. Furthermore, as the reinsurer's compensator $\varsigma^*$ is a function of $\boldalpha^*$ and $\boldeta^*$, the way the reinsurer distorts the insurer's loss model depends both on the covered loss and the safety loadings.

To illustrate the solution, we again consider the case of two insurers with both systemic and idiosyncratic losses. Each insurer's model assumes the marginal losses are exponentially distributed, but with different means and intensities.

    \begin{table}[tbh!]
        \centering
        \begin{tabular}{l@{\hskip 0.4in} l l l l l }
            Parameter & $\theta$ & $\lambda^S$ & $\xi^S$ & $\lambda^I$ & $\xi^I$\\
            \hline
            Insurer-1's model & 1.67 & 2 & 1 & 1.5 & 1.25\\
            Insurer-2's model & 1.67 & 2.5 & 1.5 & 2 & 1.5\\
        \end{tabular}
        \caption{Parameters for the numerical illustration in \cref{ex:XL_Gumbel}. The systemic losses have a Gumbel copula with parameter $\theta$ and are marginally exponentially distributed with scale parameter $\xi^S$. The idiosyncratic losses are exponentially distributed with scale parameter $\xi^I$.}
        \label{tab:XL_Gumbel_params}
    \end{table}

\begin{example}
\label{ex:XL_Gumbel}
    Suppose there are two insurers and that, under both insurers' models, the systemic losses are marginally identically exponentially distributed with scale parameter $\xi^S$ and with a dependence structure given by a Gumbel copula with parameter $\theta = 10/6$. Suppose further that under both models each insurer's idiosyncratic losses are identically exponentially distributed with scale parameter $\xi^I$. The parameters used in the numerical illustration are given in \cref{tab:XL_Gumbel_params}. Again, insurer-2's model has greater expected losses.
    
    \cref{fig:XL_2ins} shows the retention limit of each insurer at equilibrium, $\alpha_k^*$, and the safety loading charged to them, $\eta_k^*$, as a function of the weight applied by the reinsurer to insurer-2's loss model, $\pi_2$. Again, three different values of the ambiguity aversion parameter, $\eps$, are shown. Similar to the case of proportional reinsurance (\cref{ex:prop_Gamma}, in particular \cref{fig:prop_2ins}), the safety loading charged to each insurer, $\eta_k^*$, is increasing as the weight applied to the distribution with greater expected losses, $\pi_2$, increases, and the amount of reinsurance purchased by the insurers is decreasing ($\alpha_k^*$ is increasing). Furthermore, when the reinsurer is more ambiguity averse ($\eps = 0.05, 0.1$, the dashed and dotted lines compared to the solid ones), the safety loadings are increased and so is the amount retained by the insurers. 

    \begin{figure}[tbh!]
        \centering
        \scalebox{0.97}{
        \includegraphics{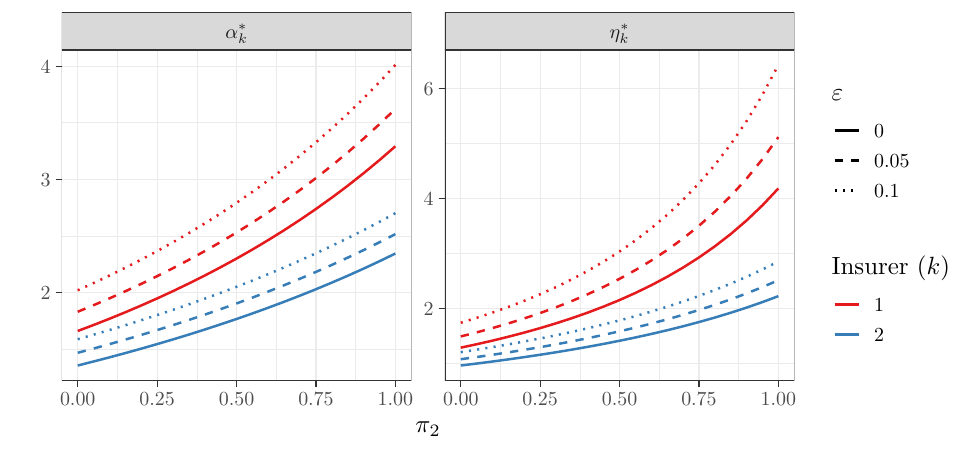}}
        \caption{The equilibrium retention limits, $\alpha_k^*$, and safety loadings, $\eta_k^*$, for insurer-$k$, $k=1,2$, as a function of the weight the reinsurer places on insurer-2's loss model, $\pi_2$. The insurers have risk aversion parameters $\gamma_1=\gamma_2=0.5$ and their severity distributions are exponentially distributed with parameters given in \cref{tab:XL_Gumbel_params}. In particular, insurer-2's model has greater expected losses.}
        \label{fig:XL_2ins}
    \end{figure}

    \cref{fig:XL_2dcomp} is a contour plot of the bivariate systemic loss density. The retention limits $\alpha^*_k$, $k=1,2$ create kinks in the joint density, which can be seen where the contour curve crosses the dashed line corresponding to the retention limit for that level of ambiguity aversion, $\eps$. Similar to the proportional reinsurance case, the reinsurer puts more weight into the tail of the joint distribution (upper right-hand corner) when they are more ambiguity averse (larger $\eps$).

    \begin{figure}[hbt!]
        \centering
        \includegraphics{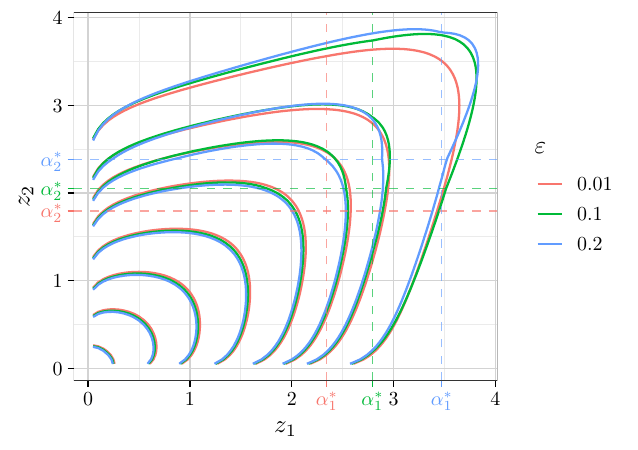}
        \caption{A contour plot of the density of the systemic loss distribution for three different values of the reinsurer's ambiguity aversion parameter, $\eps$. The retention limits of each insurer, $\alpha^*_k$, for each value of $\eps$, are marked with dashed lines. The insurers have risk aversion parameters $\gamma_1=\gamma_2=0.5$ and their severity distributions are exponentially distributed with parameters given in \cref{tab:XL_Gumbel_params}. In particular, insurer-2's model has greater expected losses.}
        \label{fig:XL_2dcomp}
    \end{figure}
\end{example}

Next, we consider the case where the excess-of-loss reinsurance is capped at a policy limit $\ell_k >0$, $k \in \N$.

\begin{proposition}[Capped excess-of-loss reinsurance]
\label{prop:XLPL}
Let $\boldalpha^*$ denote the solution of the following system of equations for $k \in \N$:
\begin{align*}
    \alpha_k^* &= \frac{1}{\gamma_k}\log
    \left(\frac{\displaystyle \int_{\alpha_k^*}^{\alpha_k^* + \ell_k}  \!\!\varrho^*_k(z_k) \, dz_k + \int_{\Rpn} \!\! \Id_{\{ \alpha_k^* \leq z_k \leq \alpha_k^* +\ell_k \}} \, \vs^*(\bz) \, d\bz}{\displaystyle \lambda_k \left[ F_k(\alpha_k^* + \ell_k) - F_k(\alpha_k^*) - \gamma_k \int_{\alpha_k^*}^{\alpha_k^* + \ell_k} \overline{F}_k(z_k) dz_k\right]}\right)\,,
\end{align*}
    where $\varrho_k^*$, $k \in \N$, and $\varsigma^*$ are given in \Cref{thm:stackelberg_eq} with $z - r(z,\alpha) = \min\{ (z-\alpha)_+, \ell_k \}$ and $\overline{F}_k(z_k) := 1 - F_k(z_k)$ is insurer-$k$'s survival function.
    If such $\boldalpha^*$ exist, then
    \begin{align*}
        \eta_k^* &= e^{\gamma_k \alpha^*_k} - 1 \,, \qquad \forall k\in\N  \,,
    \end{align*}
    and the solution to the Stackelberg game is $(\boldalpha^*,\boldeta^*,\varsigma^*, \bvrho^*)$.
\end{proposition}

We observe that as the policy limit gets large, $\ell_k \to \infty$, the retention level of the insurers, $\boldalpha^*$, approaches the retention level in \cref{prop:XL}, i.e., when there is no policy limit on the excess-of-loss reinsurance. As the safety loadings $\boldeta^*$ and the optimal compensator $\vs^*$ depend on $\boldalpha^*$, the Stackelberg equilibrium from \cref{prop:XL} (no policy limit) is the limiting case of the capped excess-of-loss reinsurance equilibrium as the policy limit goes to infinity.  Furthermore, the premium charged by the reinsurer when there is a policy limit,
\begin{equation*}
    p^R_k = 
    (1+\eta_k^*)\, \lambda_k \int_{\alpha_k^*}^{\alpha_k^* + \ell_k} \overline{F}_k(z_k) dz_k  \, ,
\end{equation*}
approaches the premium charged when the reinsurance coverage is not capped as $\ell_k \to \infty$.

  To illustrate the solution and compare the excess-of-loss policies with and without a policy limit, we return to the simpler setting where the systemic losses are comonotonic --- the model for the systemic loss is a function of a single scalar variable --- and there are no idiosyncratic losses.

\begin{example}
    \label{ex:XLPL}
    Suppose there are two insurers and the systemic losses are comonotonic and identically marginally distributed under both insurers' models. Each insurer assumes that the marginal losses are exponentially distributed with scale parameters $\xi_k\in(0,\min \{\frac{1}{\gamma_k},\,\frac{1}{2\eps} \})$, and intensity $\lambda_k>0$, for models $k=1,2$. To focus our analysis on the systemic model, we assume there are no idiosyncratic losses. We consider excess-of-loss contracts with a policy limit for each insurer $\ell_k >0$ for $k=1,2$.

    In the following numerical example, the insurers' scale parameters and intensities are $\xi_1 = 1, \, \lambda_1 = 2$ and $\xi_2 =1.25, \, \lambda_2 = 2.5$, under models 1 and 2, respectively. In particular, insurer-2 has greater expected losses than insurer-1, as it has on average more frequent claims, and the mean of their claim distribution is larger. Both insurers have the same risk aversion parameter $\gamma_1=\gamma_2=0.5$.

    \begin{figure}[hbt!]
    \centering
    \subfloat[Reinsurer's severity distribution, no policy limit]{
    \scalebox{0.85}{\includegraphics{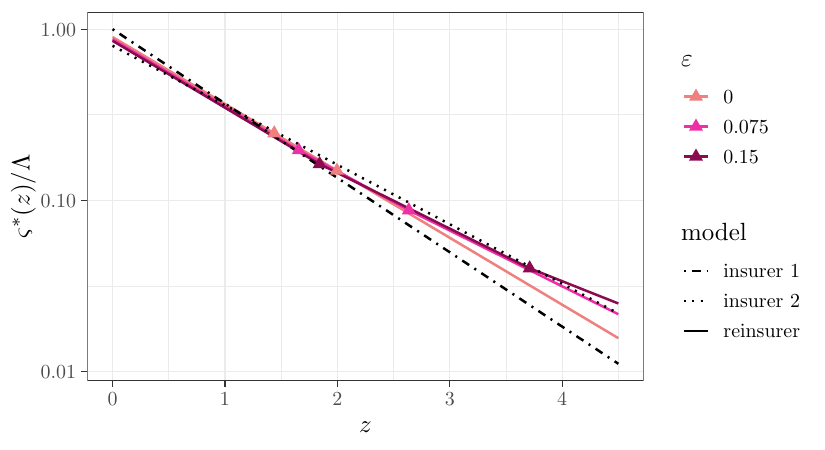}} \label{fig:1d_comp_XL}}
    \hspace{1em}
    \subfloat[Reinsurer's severity distribution, policy limit $\ell_1=\ell_2=1$]{
    \scalebox{0.85}{\includegraphics{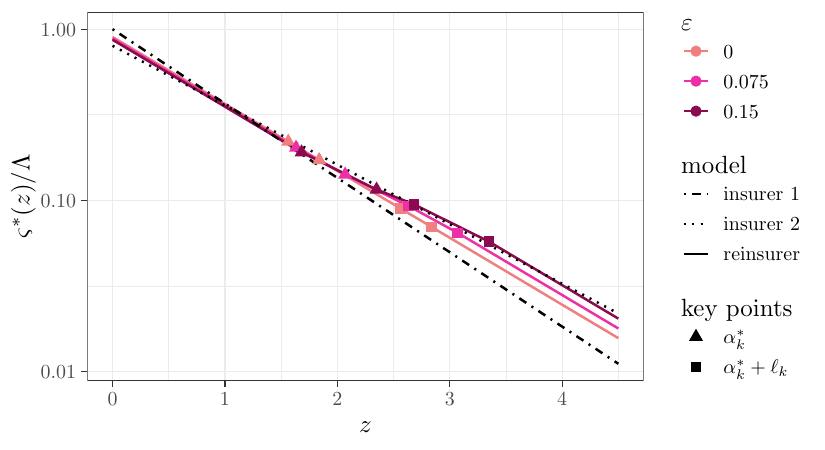}}\label{fig:1d_comp_XLPL}}
    \caption{The severity distributions of the systemic loss on a log scale under insurer-1's model (dot-dashed line), insurer-2's model (dotted line), and the reinsurer's model (solid lines) for three values of $\eps$. The triangles mark the retention limits of each insurer, $\alpha_1^*$ and $\alpha_2^*$, and in panel (b), the squares mark the upper limit of the reinsurance coverage, $\alpha_1^* + \ell_1$ and $\alpha_2^* + \ell_2$.
    The insurers have risk aversion parameters $\gamma_1=\gamma_2=0.5$, intensities $\lambda_1 = 2,  \, \lambda_2 = 2.5$, and the systemic losses are comonotonic under both models, with identical exponentially distributed marginals with scale parameters $\xi_1 = 1, \, \xi_2 =1.25$. In particular, insurer-2's model has greater expected losses. The insurers' models are equally weighted with $\pi_1 = \pi_2 =0.5$.}
    \label{fig:XL_comps_como}
  \end{figure}

    \cref{fig:1d_comp_XL} plots the severity distributions of insurer-1 (dot-dashed line), insurer-2 (dotted line), and the reinsurer (solid lines) on a log scale in the case with no policy limit. We assume that the insurers' models are equally weighted by the reinsurer, i.e., $\pi_1 = \pi_2 =0.5$. The retention limits of each insurer, $\alpha_1^*$ and $\alpha_2^*$, are marked with triangles. We consider three levels of ambiguity-aversion for the reinsurer: $\eps=0,\,0.075, \, 0.15$, where the first case corresponds to the reinsurer being ambiguity-neutral. When $\eps=0$, the reinsurer's severity distribution is the normalised weighted geometric mean of the insurers' severity distributions, which appears as a straight line on the log scale as it is exponentially distributed. Including ambiguity aversion shifts the tail of the reinsurer's severity distribution up, creating kinks at $z=\alpha_1^*$ and $z= \alpha_2^*$. Thus, when the reinsurer is ambiguity-averse, they add weight to the parts of the distribution that are past each insurer's retention limit. Note that for all $\eps$, $\alpha_2^*  < \alpha^*_1$, and for larger ambiguity of the reinsurer, we observe that $\alpha_k^*$ shifts to the right, indicating that both insurers buy less reinsurance. 
    
    \cref{fig:1d_comp_XLPL} shows the reinsurer's severity distribution on the log scale when the policy limits are $\ell_1 = \ell_2 = 1$, which corresponds to the 63rd and 55th quantiles of the insurer-1 and insurer-2' loss distributions, respectively, prior to purchasing reinsurance. We assume both insurers' models are equally weighted ($\pi_1 = \pi_2 =0.5$). Similar to \cref{fig:1d_comp_XL}, at the retention limits, $\alpha_k^*$ (marked with triangles), the compensator shifts away from the normalised geometric mean when the reinsurer is ambiguity averse ($\eps > 0$). 
    However, once the policy limits of both insurers are reached (at $\alpha_k^* + \ell_k$, $k=1,2$, marked with squares), the severity distribution again has an exponential decay.

    \cref{fig:XLPL_vars} shows the insurers' retention limits, $\alpha_k^*$, their safety loadings, $\eta_k^*$, as well as the premium charged, $p_k^R$, as a function of the policy limit, $\ell_k$. The faint horizontal lines give the optimal values when there is no policy limit. As expected, we observe that in all cases the variables approach the no policy limit equilibrium as $\ell_k$ gets large --- note that $\ell_k = 6$ corresponds approximately to the 99th quantile of both insurers' loss distributions.

    \begin{figure}[htb!]
        \centering
        \scalebox{0.9}{\includegraphics{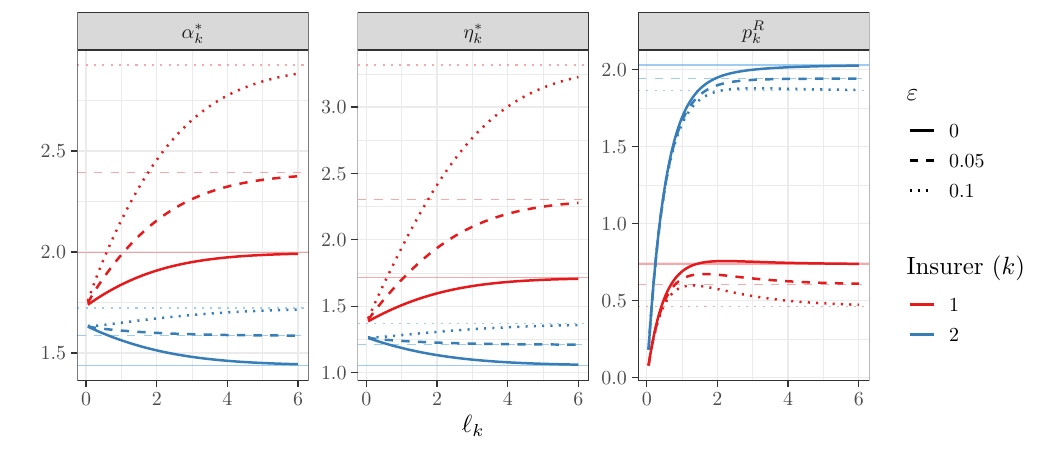}}
        \caption{The retention limit, $\alpha_k^*$, the safety loading, $\eta_k^*$, and the premium, $p_k^R$, charged to each insurer as a function of the policy limit, $\ell_k$. The dotted horizontal lines give the corresponding values for the case of no policy limit. The insurers have risk aversion parameters $\gamma_1=\gamma_2=0.5$, intensities $\lambda_1 = 2,  \, \lambda_2 = 2.5$, and the systemic losses are comonotonic under both models, with identical exponentially distributed marginals. The exponential scale parameters are $\xi_1 = 1, \, \xi_2 =1.25$ under the models of insurer-1 and insurer-2, respectively.}
        \label{fig:XLPL_vars}
    \end{figure}
    
    First, examine the equilibrium values for insurer-1 (red lines). Their retention limit, $\alpha_1^*$, is increasing in $\ell_1$ for all levels of ambiguity aversion, $\eps$. Moreover, the safety loading, $\eta_1^*$, is also increasing in $\ell_1$ for all $\eps$. The premium, $p^R_1$, however, is not monotone in $\ell_1$, as it depends on both the reinsurance purchased and the safety loading. Indeed, we observe that for large ambiguity, $\eps = 0.1$, the safety loading is steeply increasing in $\ell_1$, which results that insurer-1's demand of reinsurance drastically decreases.

    The equilibrium values for the riskier insurer (insurer-2, blue lines) have different behaviour with increasing $\ell_2$: when the reinsurer is not ambiguity-averse, insurer-2 decreases their retention limit as the policy limit increases, purchasing more reinsurance. Moreover, the safety loading is also decreasing in $\ell_2$, however, the combined result is that insurer-2 pays a larger premium, $p_2^R$, as $\ell_k$ increases. This behaviour is dampened when the reinsurer is ambiguity-averse ($\eps =0.1$), as the safety loading $\eta_2^*$ becomes increasing in $\ell_2$.
\end{example}

\section{Reinsurer's utility maximisation strategy under ambiguity}
\label{sec:utility_max}

In the literature, the problem of a reinsurer who maximises their expected utility, rather than their expected wealth, has been well-studied for the case of one insurer. Examples include \textcite{Yi2013}, \textcite{YiViensLiZeng2015}, \textcite{Hu2018continuoustime}, \textcite{Hu2018}, and \textcite{Li2018}. In this section, we compare our problem to this related utility-maximisation problem.
For this, we assume that the reinsurer, like the insurers, has a CARA utility function $u(x) = -\frac{1}{m} e^{-m x}$ with risk aversion parameter $m > 0$. To obtain an analytical solution, we follow the convention in the literature and adopt a state-dependent ambiguity preference term for the reinsurer.  First introduced by \textcite{Maenhout2004}, one generalises the ambiguity aversion parameter $\eps$ to one that depends on time and the reinsurer's wealth, $(t,y)$, as follows:
\begin{equation*}
    \varphi(t,y) := - \frac{\eps}{m \, \Phi (t,y)} \, ,
\end{equation*}
where $\Phi(\cdot,\cdot)$ is the value function of the reinsurer's utility maximisation problem.
This type of ambiguity aversion has been used extensively in reinsurance games in continuous time, and the result is that the reinsurer's ambiguity aversion scales inter-temporally with the value function. These changes give the following modified value function for the reinsurer:
\begin{equation}
    \begin{split}
    \Phi(t,y) &= \sup_{\eta\in\mfC} \inf_{\boldsymbol{\sigma} \in \mathcal{V}} \, \mathbb{E}^{\Q^{\boldsymbol{\sigma}}}_{t,y}  \Bigg[ u(Y_{T}) +   \sum_{j \in\N} \pi_{j} \Bigg(\, \int_t^T \!\! \int_{\Rpn}  \frac{1}{\varphi(u,Y_{u-})} \!\left( \varsigma_u(\bz) \left[ \log \left( \frac{\varsigma_u(\bz)}{v_{j}(\bz)}\right) - 1 \right] + v_{j}(\bz) \right)  d\bz\,  du \\
    & \qquad\qquad\qquad\quad + \sum_{k\in\N} \int_t^T \!\! \int_{\R_+} \!\!  \frac{1}{\varphi(u,Y_{u-})} \left(\varrho_{u,k}(z_k) \left[ \log \left( \frac{\varrho_{u,k}(z_k)}{w_{k,j}(z_k)}\right) - 1\right] + w_{k,j} (z_k) \right) dz_k \,du  \Bigg) \Bigg] \,.
    \label{eq:modified_valuefunc}
    \end{split}
\end{equation}

The solution to the reinsurer's problem with this modified value function is given in the following proposition. We assume the insurers' problem is unchanged, so the solution to their problem is the one given in \cref{prop:insurer_optim}. 

\begin{proposition}
\label{prop:scaled_prob}
    For $\boldsymbol{c}:=(c_1, \ldots, c_n)\in\R^n$, define the following parameterised set of compensators: 
    \begin{align*}
	   \varrho_k^\dagger(z_k, c_k) & :=  w_k^g(z_k) \exp \left\{ \frac{\eps}{m} \left[ \exp \left( {m[z_k - r(z_k, \alpha_k^\dagger[c_k] )]} \right) - 1 \right]\right\} \,, \quad \forall k \in \N \,, \text{ and} \\
	   \varsigma^\dagger(\bz,\boldsymbol{c}) 
        :&= 
        v^g(\bz)\; \exp \left\{ \frac{\eps}{m} \left[ \exp \left(m{\sum_{k\in\N} \left[ z_k - r(z_k,  \alpha_k^\dagger[c_k]) \right]} \right) -1 \right]  \right\} \,.
    \end{align*}
    Further, suppose that the system of non-linear equations in $c_k$, $k\in\N$,
    \begin{equation}
        \begin{split}
            1+c_k = &
            \frac{\displaystyle \int_{\R_+} \!\!\! \left(z_k - r(z_k,\alpha_k^\dagger[c_k] )\right) F_k(dz_k) }{ \displaystyle \int_{\R_+} \!\!\!  \partial_c \, r(z_k,\alpha_k^\dagger[c])\big|_{c=c_k} F_k(dz_k) } + \frac{\displaystyle \int_{\R_+} \!\!  e^{m[z_k - r(z_k, \alpha_k^\dagger[c_k] )]} \partial_c \, r(z_k,\alpha_k^\dagger[c])\big|_{c=c_k} \, \varrho_k^\dagger(z_k, c_k) \, dz_k}{ \displaystyle \lambda_k \int_{\R_+} \!\!\!  \partial_c \, r(z_k,\alpha_k^\dagger[c])\big|_{c=c_k} F_k(dz_k) } \\
            &+ \frac{\displaystyle \int_{\Rpn} \! e^{m \sum_{k\in\N} \left[ z_k - r(z_k,  \alpha_k^\dagger[c_k]) \right]} \partial_c \, r(z_k,\alpha_k^\dagger[c])\big|_{c=c_k} \, \varsigma^\dagger(\bz,\boldsymbol{c}) \, d\bz}{ \displaystyle \lambda_k \int_{\R_+} \!\!\!  \partial_c \, r(z_k,\alpha_k^\dagger[c])\big|_{c=c_k} F_k(dz_k) }
            \label{eqn:ck_modified}
        \end{split}
        \raisetag{\normalbaselineskip}
    \end{equation}
    has a solution, which we denote by $\boldsymbol{\eta}^*$.
    Then, the reinsurer's optimal compensators in feedback form for \eqref{eq:modified_valuefunc}  are $\bvrho^*(\cdot) = (\varrho_{1}^\dagger(\cdot, \eta^*_1),\ldots,\varrho^\dagger_{n}(\cdot, \eta_n^*))$ and $\varsigma^*(\cdot):=\varsigma^\dagger(\cdot,\boldsymbol{\eta}^*)$.
\end{proposition}

The proof of this proposition is similar to the proof of \cref{prop:reinsurer_optim} and can be found in \cref{sec:app}.
Combining this result with \cref{prop:insurer_optim}, we obtain the Stackelberg equilibrium for the case where the reinsurer has exponential utility and a scaled entropy penalty.

\begin{theorem}
Under the assumptions of \cref{prop:insurer_optim} and \cref{prop:scaled_prob}, the Stackelberg equilibrium when the reinsurer maximises their exponential utility is $(\boldalpha^*, \boldeta^*, \varsigma^*,\boldsymbol{\varrho}^*)$, where $\boldalpha^* = (\alpha^*_1, \ldots, \alpha^*_n)$, with $\alpha^*_k:= \alpha^\dagger_k[\eta_k^*]$, $k \in \N$, $\boldsymbol{\varrho}^* = (\varrho_1^*(\cdot), \ldots, \varrho_n^*(\cdot))$ where for each $k \in \N$,
\begin{equation*}
    \varrho_k^*(z_k) =  w_k^g(z_k) \exp \left\{ \frac{\eps}{m} \left[ \exp \left( {m[z_k - r(z_k, \alpha_k^*)]} \right) - 1 \right]\right\} \, ,
\end{equation*}
and
\begin{equation*}
    \varsigma^*(\bz) = v^g(\bz)\; \exp \left\{ \frac{\eps}{m} \left[ \exp \left(m{\sum_{k\in\N} \left[ z_k - r(z_k,  \alpha_k^*) \right]} \right) -1 \right]  \right\} \, .
\end{equation*}
    The $\boldeta^* = (\eta_1^*, \ldots, \eta_n^*)$ satisfy the system of equations
    \begin{align*}
        0 &= \int_{\R_+}  \!\! \partial_a \, r(z_k,a)\big|_{a=\alpha^\dagger[\eta_k^*]} \left\{(1+\eta_k^*)-e^{\gamma_k r(z_k,\alpha^\dagger[\eta_k^*])}\right\} F_k(dz_k) \,, \quad \forall k\in\N \, ,  \\
        1 + \eta_k^* &= \frac{\displaystyle \int_{\R_+} \!\!\! \left(z_k - r(z_k,\alpha_k^* )\right) F_k(dz_k) }{ \displaystyle \int_{\R_+} \!\!\!  \partial_c \, r(z_k,\alpha_k^\dagger[c])\big|_{c=\eta_k^*} F_k(dz_k) } +  \frac{\displaystyle\int_{\R_+} \!\!  e^{m[z_k - r(z_k, \alpha_k^* )]} \partial_c \, r(z_k,\alpha_k^\dagger[c])\big|_{c=\eta_k^*} \, \varrho_k^*(z_k) \, dz_k}{ \displaystyle \lambda_k \int_{\R_+} \!\!\!  \partial_c \, r(z_k,\alpha_k^\dagger[c])\big|_{c=\eta_k^*} F_k(dz_k) } \\
        &+ \frac{\displaystyle \int_{\Rpn} \! e^{m \sum_{k\in\N} \left[ z_k - r(z_k,  \alpha_k^*) \right]} \partial_c \, r(z_k,\alpha_k^\dagger[c])\big|_{c=\eta_k^*} \, \varsigma^*(\bz) \, d\bz}{ \displaystyle \lambda_k \int_{\R_+} \!\!\!  \partial_c \, r(z_k,\alpha_k^\dagger[c])\big|_{c=\eta_k^*} F_k(dz_k) }
        \,, \quad \forall k\in\N   \,. 
    \end{align*}
\end{theorem}

The main difference to the case where the reinsurer maximises their expected wealth (\cref{thm:stackelberg_eq}) is the reinsurer's compensators. Here, for both the compensators for the idiosyncratic and systemic losses, the weighted geometric mean of the insurers' compensators is multiplied by a double exponential term which depends on the (aggregate) loss of the reinsurer, as opposed to a single exponential term. In addition, the second set of equations has an additional exponential term in the second and third integrals in the numerator. 

Comparing the compensators $\varsigma^*$, $\varrho^*_k$, $k\in\N$, from the two problems, we observe that the compensators in the utility-maximisation case are greater than or equal to the respective compensators of the wealth-maximisation case for all $m>0$ and any non-negative loss. We also observe that as $m \to 0$, that is, diminishing absolute risk aversion, the compensators approach those in the wealth-maximising case. Furthermore, the entire Stackelberg equilibrium approaches the one where the reinsurer maximises their expected wealth (\cref{thm:stackelberg_eq}) as $m \to 0$.

Next, we observe that when the reinsurer is both risk-averse (CARA utility, $m>0$) and ambiguity-averse ($\eps > 0$), they adopt a model that puts more weight on tail losses. Furthermore, as $z_k - r(z_k,\alpha^*)$ is a function of $z_k$, its double exponential is not integrable for many common types of reinsurance contracts and common choices of loss distributions. For example, excess-of-loss reinsurance without a policy limit and proportional reinsurance require the existence of the moment generating function of the random variable $e^Z$, where $Z$ is the aggregate loss of the reinsurer. Thus, in these cases, the reinsurance contracts cannot be priced for the case when the reinsurer is both ambiguity- and risk- averse. 

Moreover, our solution is consistent with the literature in the following sense. For example, \textcite{Hu2018} allow for measure changes which depend on time only and find that for an excess-of-loss reinsurance contract, the optimal probability distortion is
\begin{equation*}
    \varsigma^* = v(z) \exp \left(\frac{\eps}{m} \left[ \left( \int_{\alpha^*}^\infty e^{m(z-\alpha^*)} - 1  \right) F(dz) \right] \right)
\end{equation*}
(adjusting for notational differences and differences in investment behaviour of the reinsurer). In our setting, allowing for time and state dependence and in the special case of a single insurer with a single source of loss and an excess-of-loss contract, we obtain
\begin{equation*}
    \varsigma^*(z) = v(z)  \exp \left( \frac{\eps}{m} \left[ e^{m  (z - \alpha^*)} - 1 \right] \right) \, .
\end{equation*}
Thus, the difference is an expectation within the exponential. Allowing the probability distortion to depend on the loss size (equivalently, allowing the severity distribution to change under the optimal measure) generalises the results of \textcite{Hu2018}. It also illustrates an important point: the problem where the reinsurer is both risk-averse and ambiguity-averse is ill-posed when the reinsurer's probability distortion can depend on the size of the loss $z$, as the reinsurer seeks to avoid risk to an extent to which they cannot price their reinsurance contracts. Thus, it is preferable to consider the case of a wealth-maximising reinsurer, as in this paper, to not penalise tail losses twice.

\section{Conclusion}
We study optimal reinsurance in continuous time via a Stackelberg game. Specifically, we consider the novel problem of an ambiguity-averse reinsurer who interacts with multiple insurers, where each insurer has a different loss model and utility function. The insurable losses consist of systemic losses, which affect all insurers, and idiosyncratic losses, which affect insurers individually. Each insurer chooses the portion of their losses to reinsure that optimises their expected utility, while the reinsurer sets prices that maximise their expected wealth and simultaneously accounts for ambiguity around the insurers' loss models. In particular, the reinsurer has (a) a priori beliefs in each insurer's model, and then penalises each model with the Kullback-Leibler divergence, and (b) an overall ambiguity factor indicating uncertainty in the insurance market. We solve this Stackelberg game for a unified framework of reinsurance contracts and allow the reinsurer's probability distortion to depend on time, claim frequency, and claim severity. We find that at the Stackelberg equilibrium, the reinsurer sets prices under a model which is closely related to the barycentre of the insurers' loss models. Indeed if the reinsurer is ambiguity neutral, they choose the barycentre model, while if the reinsurer is ambiguity averse, they weight the tail of their losses more heavily. We illustrate the Stackelberg equilibrium on proportional reinsurance and excess-of-loss reinsurance with and without a policy limit.

We further compare our solution to the case where the ambiguity-averse reinsurer maximises their CARA utility, generalising known results in the literature to account for ambiguity in the severity distribution. We find that a solution to the utility maximisation problem only exists in very restrictive cases, such as when the losses covered by the reinsurer are bounded. In this case, the safety loadings charged to the insurers are increasing in the reinsurer's risk aversion parameter.

\section*{Acknowledgements}
EK is supported by an NSERC Canada Graduate Scholarship-Doctoral. SJ and SP would like to acknowledge support from the Natural Sciences and Engineering Research Council of Canada (grants RGPIN-2018-05705 and RGPIN-2024-04317,  and DGECR-2020-00333, RGPIN-2020-04289). SP also acknowledges support from the Canadian Statistical Sciences Institute (CANSSI).

\begin{appendices}

\crefalias{section}{appendix}

\section{Proof of \cref{prop:scaled_prob}}
\label{sec:app}
We write the reinsurer's value function as
\begin{align*}
    \Phi(t,y) = \sup_{\eta\in\mfC} \inf_{\boldsymbol{\sigma} \in \mathcal{V}} \, \mathbb{E}^{\Q^{\boldsymbol{\sigma}}}_{t,y}  \Bigg[ u(Y_{T}) -
    \frac{m}{\eps} &\Bigg(\, \sum_{k\in\N} \int_t^T \!\! \int_{\R_+} \!\!\! \Phi(u,Y_{u-})\! \left(\varrho_{u,k}(z_k) \left[ \log \left( \frac{\varrho_{u,k}(z_k)}{w^g_{k}(z_k)}\right) - 1\right] + w^a_{k} (z_k) \right) \! dz_k \,du \\
    & \quad\, + \int_t^T \!\! \int_{\Rpn} \!\!\! \Phi(u,Y_{u-}) \!\left( \varsigma_u(\bz) \left[ \log \left( \frac{\varsigma_u(\bz)}{v^g(\bz)}\right) - 1 \right] + v^a(\bz) \right)  d\bz\,  du  \Bigg) \Bigg] \,.
\end{align*}
By the dynamic programming principle, the value function $\Phi$  satisfies the HJBI equation:
\begin{align*}
    0 = \partial_t \Phi(t,y) + \sup_{\boldsymbol{c}\in\R^n_+} \inf_{\substack{\psi \in \G, \\ \xi_1, \ldots, \xi_n \in \mathcal{H}}}  \Bigg\{ & \sum_{k\in\N} p_k^R\big(\alpha^\dagger_k[c_k], c_k \big) \partial_y \Phi(t,y) \\
    &+ \int_{\Rpn}  \left[ \Phi\left(t,y-\sum_{k\in\N} \big[z_k-r(z_k,\alpha^\dagger_k[c_k])\big] \right) \,  - \Phi(t,y) \right] \psi(t,\bz) \, d\bz 
    \\
    &+  \sum_{k\in\N} \int_{\R_+}  \left[ \Phi\left(t,y- \big[z_{k}-r(z_{k},\alpha^\dagger_k[c_k])\big] \right) \,  - \Phi(t,y) \right] \xi_k(t,z_k) \, dz_k
    \\
    & -  \left( \frac{m}{\eps} \right) \Phi(t,y) \int_{\Rpn} \left( \psi(t,\bz) \left[ \log \left( \frac{\psi(t,\bz)}{v^g(\bz)}\right) - 1 \right] + v^a(\bz) \right)  d\bz \\
    & - \left( \frac{m}{\eps} \right) \Phi(t,y) \sum_{k\in\N} \int_{\R_+} \left( \xi_{k}(t,z_k) \left[ \log \left( \frac{\xi_{k}(t,z_k)}{w_{k}^g(z_k)}\right) - 1 \right] + w_{k}^a (z_k) \right)  dz_k \Bigg\} \, 
\end{align*}
with terminal condition $\Phi(T,y) = -\tfrac{1}{m} e^{-my}$.

Given the functional form of the terminal condition, we make the Ansatz $\Phi(t,y) = -\tfrac{1}{m}\,e^{-m[f(t) + g(t) y]}$
for deterministic functions $f(t)$ and $g(t)$ satisfying terminal conditions $ f(T) = 0$ and $g(T) = 1$. Upon substituting the Ansatz into the HJBI equation and simplifying, we obtain the ODE
\begin{align*}
    0 = mf'(t) + \sup_{\boldsymbol{c}\in\R^n_+} \inf_{\substack{\psi \in \G, \\ \xi_1, \ldots, \xi_n \in \mathcal{H}}} \Bigg\{ & m \sum_{k\in\N} p_k^R\left(\alpha^\dagger_k[c_k], c_k \right)
    + \, \frac{m}{\eps} \left(\int_{\Rpn} \!\! v^a(\bz) d\bz + \int_{\R_+} w^a_k(z_k) dz_k  \right) \\
    & - \int_{\Rpn} \psi(t,\bz) \left[ \frac{m}{\eps}  \left( \log \left( \frac{\psi(t,\bz)}{v^g(\bz)}\right) - 1 \right) + 1 - \exp\left( \sum_{k\in\N} m \big[z_k-r(z_k,\alpha^\dagger_k[c_k])\big] \right) \right]  d\bz \\
    & -\sum_{k\in\N} \int_{\R_+} \!\!\! \xi_{k}(t,z_k) \left[ \frac{m}{\eps} \left( \log \left( \frac{\xi_{k}(t,z_k)}{w_{k}^g(z_k)}\right) - 1 \right) + 1 - \exp\left( m \big[z_k-r(z_k,\alpha^\dagger_k[c_k])\big] \right) \right]  dz_k
    \Bigg\} \, , 
\end{align*}
with terminal condition $ f(T) = 0$ and that $g(t)=1$ for all $t\in[0,T]$.

We first consider the inner optimisation and in a first step optimise over $\psi$, and in a second step over $\xi_1, \ldots, \xi_n$. The infimum problem, where we minimise the functional of $\psi$, is given by
\begin{equation*}
    \mathcal{L}[\psi] = \int_{\Rpn} \psi(t,\bz) \left[ \frac{m}{\eps}  \left( \log \left( \frac{\psi(t,\bz)}{v^g(\bz)}\right) - 1 \right) + 1 - \exp\left( \sum_{k\in\N} m \big[z_k-r(z_k,\alpha^\dagger_k[c_k])\big] \right) \right]  d\bz \, .
\end{equation*}

Let $\delta >0$ and $h$ be an arbitrary function such that $\delta h + \psi \in \G$. Applying a variational first order condition, we obtain
\begin{align*}
    0 &= \lim_{\delta \to 0} \frac{\mathcal{L}[ \psi + \delta\,h] - \mathcal{L}[\psi]  }{\delta} = \int_{\Rpn} h(\bz) \left[\frac{m}{\eps} \log \left( \frac{\psi(t,\bz)}{v^g(\bz)}\right) + 1 - \exp\left( \sum_{k\in\N} m \big[z_k-r(z_k,\alpha^\dagger_k[c_k])\big] \right) \right] d\bz \, .
\end{align*}
As $h$ was arbitrary, this implies that the term in the square brackets must vanish; therefore, the optimal compensator in feedback form, $\vs^*$, is
\begin{equation*}
    \varsigma^*_t(\bz, \boldsymbol{c} ) = v^g(\bz)\, \exp \left\{ \frac{\eps}{m} \left[ \exp \left(m{\sum_{k\in\N} \left[ z_k - r(z_k,  \alpha_k^\dagger[c_k]) \right]} \right) -1 \right]  \right\} \,. 
\end{equation*}

In the second step, we use a similar procedure to solve individually for each $\xi_k$, $k\in\N$, giving each optimal compensator in feedback form, i.e. $\varrho^*_k$:
\begin{equation*}
    \varrho^*_{t,k}(z_k,c_k) = w_k^g(z_k) \exp \left\{ \frac{\eps}{m} \left[ \exp \left( {m\,[z_k - r(z_k, \alpha_k^\dagger[c_k] )]} \right) - 1 \right]\right\} \, .
\end{equation*}
Substituting this into the HBJI equation above, we have that
\begin{equation*}
        0 = m f'(t) + \sup_{\boldsymbol{c}\in\R^n_+} \Bigg\{  m \sum_{k\in\N} p_k^R\big(\alpha^\dagger_{k}[c_k], c_k\big) + \, \frac{m}{\eps} \left(\int_{\Rpn} \!\! \left[ v^a(\bz) -  \varsigma^*_t(\bz, \boldsymbol{c}) \right] d\bz + \int_{\R_+} \left[w^a_k(z_k) - \varrho^*_{t,k}(z_k,c_k) \right]dz_k  \right) \Bigg\} \, .
\end{equation*}
Next, the derivative of the term in curly braces with respect to $c_k$, $k \in \N$ is 
\begin{align*}
    \frac{\partial}{\partial c_k} \{ \cdots \} =&
     m \, \lambda_k \left[ \int_{\R_+} \!\!\! \left[ z_k - r(z_k,\alpha_k^\dagger[c_k] )\right] F_k(dz_k)
    - m (1+c_k) \! \int_{\R_+} \!\!\!  \partial_{c_k} r(z_k,\alpha_k^\dagger[c_k])  F_k(dz_k) \right]
    \\
    & + m \int_{\R_+} e^{m \left[ z_k - r(z_k,  \alpha_k^\dagger[c_k]) \right]} \,\partial_{c_k}r(z_k,\alpha_k^\dagger[c_k]) \, \varrho_{t,k}^*(z_k, c_k) \,dz_k \\
    & +  m \int_{\Rpn} e^{m \sum_{k\in\N}  \left[ z_k - r(z_k,  \alpha_k^\dagger[c_k]) \right]} \,
    \partial_{c_k}r(z_k,\alpha_k^\dagger[c_k]) \,\varsigma^*_t(\bz, \boldsymbol{c}) \,d\bz\,.
\end{align*}
The first order conditions require setting this expression to $0$ for all $k\in\N$ and all $t\in[0,T]$. Hence, we obtain that the optimal $\eta_{k,t}^*$ satisfy the non-linear equations given in \eqref{eqn:ck_modified}. As there is no explicit time dependence in the above equation, the roots of these non-linear equations are constant over time. \qed

\end{appendices}
\printbibliography

@book{Mikosch2009,
  title     = {Non-Life Insurance Mathematics},
  subtitle  = {An Introduction with the {P}oisson Process},
  author    = {Mikosch, Thomas},
  year      = {2009},
  publisher = {Springer Berlin, Heidelberg},
  doi       = {10.1007/978-3-540-88233-6}
}

@book{appliedstochjump,
  title={Applied Stochastic Control of Jump Diffusions},
  author={{\O}ksendal, Bernt and Sulem, Agn\`{e}s},
  year={2019},
  edition = {3},
  series = {Universitext},
  publisher={Springer},
  doi = {10.1007/978-3-540-69826-5}
}

@book{Albrecher2017,
  title     = {Reinsurance: Actuarial and Statistical Aspects},
  author    = {Albrecher, Hansj\"{o}rg and Beirlant, Jan and Teugels, Jozef L.},
  year      = {2017},
  publisher = {John Wiley \& Sons Ltd}
}

@book{Asmussen2020,
  title     = {Risk and Insurance},
  subtitle  = {A Graduate Text},
  author    = {Asmussen, S{\o}ren and Steffensen, Mogens},
  year      = {2020},
  publisher = {Springer Cham},
  doi={10.1007/978-3-030-35176-2}
}

@article{Yi2013,
title = {Robust optimal control for an insurer with reinsurance and investment under Heston’s stochastic volatility model},
journal = {Insurance: Mathematics and Economics},
volume = {53},
number = {3},
pages = {601-614},
year = {2013},
%issn = {0167-6687},
doi = {10.1016/j.insmatheco.2013.08.011},
%url = {https://www.sciencedirect.com/science/article/pii/S0167668713001327},
author = {Bo Yi and Zhongfei Li and Frederi Viens and Yan Zeng},
keywords = {Reinsurance and investment strategy, Stochastic volatility, Robust optimal control, Utility maximization, Ambiguity-Averse Insurer}
}

@article{YiViensLiZeng2015,
author = {Bo Yi and Frederi Viens and Zhongfei Li and Yan Zeng},
title = {Robust optimal strategies for an insurer with reinsurance and investment under benchmark and mean-variance criteria},
journal = {Scandinavian Actuarial Journal},
volume = {2015},
number = {8},
pages = {725-751},
year  = {2015},
publisher = {Taylor \& Francis},
doi = {10.1080/03461238.2014.883085}
}

@article{Li2018,
author = {Danping Li and Yan Zeng and Hailiang Yang},
title = {Robust optimal excess-of-loss reinsurance and investment strategy for an insurer in a model with jumps},
journal = {Scandinavian Actuarial Journal},
volume = {2018},
number = {2},
pages = {145-171},
year  = {2018},
publisher = {Taylor & Francis},
doi = {10.1080/03461238.2017.1309679}}

@article{Irgens2004,
title = {Optimal control of risk exposure, reinsurance and investments for insurance portfolios},
journal = {Insurance: Mathematics and Economics},
volume = {35},
number = {1},
pages = {21-51},
year = {2004},
%issn = {0167-6687},
doi = {10.1016/j.insmatheco.2004.04.004},
%url = {https://www.sciencedirect.com/science/article/pii/S0167668704000496},
author = {Christian Irgens and Jostein Paulsen},
keywords = {Optimal control, Diffusion perturbated risk process, Hamilton–Jacobi–Bellmann equation, Proportional reinsurance, Excess of loss reinsurance, Investment strategy}
}

@article{Liang2011,
title = {Optimal proportional reinsurance and investment in a stock market with Ornstein–Uhlenbeck process},
journal = {Insurance: Mathematics and Economics},
volume = {49},
number = {2},
pages = {207-215},
year = {2011},
%issn = {0167-6687},
doi = {10.1016/j.insmatheco.2011.04.005},
%url = {https://www.sciencedirect.com/science/article/pii/S016766871100045X},
author = {Zhibin Liang and Kam Chuen Yuen and Junyi Guo},
keywords = {Stochastic control, Hamilton–Jacobi–Bellman equation, Ornstein–Uhlenbeck process, Compound Poisson process, Brownian motion, Exponential utility, Filtering, Partial observations, Proportional reinsurance, Investment}
}

@article{Maenhout2004,
    author = {Maenhout, Pascal J.},
    title = "{Robust Portfolio Rules and Asset Pricing}",
    journal = {The Review of Financial Studies},
    volume = {17},
    number = {4},
    pages = {951-983},
    year = {2004},
    month = {04},
%    issn = {0893-9454},
    doi = {10.1093/rfs/hhh003}
%    url = {10.1093/rfs/hhh003},
%    eprint = {https://academic.oup.com/rfs/article-pdf/17/4/951/24435962/hhh003.pdf},
}

@article{ChenShen2018,
title={On a new paradigm of optimal reinsurance: a stochastic Stackelberg differential game between an insurer and a reinsurer},
author={Chen, Lv and Shen, Yang},
volume={48},
DOI={10.1017/asb.2018.3},
number={2},
journal={ASTIN Bulletin},
publisher={Cambridge University Press},
year={2018},
pages={905–960}}

@article{ChenShen2019,
title = {Stochastic Stackelberg differential reinsurance games under time-inconsistent mean–variance framework},
journal = {Insurance: Mathematics and Economics},
volume = {88},
pages = {120-137},
year = {2019},
issn = {0167-6687},
doi = {10.1016/j.insmatheco.2019.06.006},
% url = {https://www.sciencedirect.com/science/article/pii/S0167668718301835},
author = {Lv Chen and Yang Shen},
keywords = {Leader–follower, Proportional reinsurance, Excess-of-loss reinsurance, Mean–variance, Time inconsistency, Stackelberg equilibrium}
}

@article{Hu2018continuoustime,
author = {Duni Hu and Shou Chen and Hailong Wang},
title = {Robust reinsurance contracts in continuous time},
journal = {Scandinavian Actuarial Journal},
volume = {2018},
number = {1},
pages = {1-22},
year  = {2018},
publisher = {Taylor & Francis},
doi = {10.1080/03461238.2016.1274270}
}

@article{Hu2018,
title = {Robust reinsurance contracts with uncertainty about jump risk},
journal = {European Journal of Operational Research},
volume = {266},
number = {3},
pages = {1175-1188},
year = {2018},
%issn = {0377-2217},
doi = {10.1016/j.ejor.2017.10.061},
%url = {https://www.sciencedirect.com/science/article/pii/S0377221717309906},
author = {Duni Hu and Shou Chen and Hailong Wang},
keywords = {Game theory, Ambiguity, Proportional reinsurance, Excess-loss reinsurance, Reinsurance price},
abstract = {We investigate robust reinsurance contracts in two reinsurance modes, namely proportional reinsurance and excess-loss reinsurance, in a continuous-time principal–agent framework. Insurance claims follow the classic Cramer–Lundberg process. The reinsurer (principal) is concerned about potential ambiguity in the claim intensity, but the insurer (agent) is not. The reinsurer designs a robust reinsurance contract that maximizes the penalty-based multiple-priors utility of terminal wealth, subject to the insurer’s incentive compatibility constraint. We derive the analytical expressions of the robust reinsurance contacts. Our results show that the reinsurer dynamically decreases the reinsurance price, which makes the demand for reinsurance increase over time. However, the reinsurer’s ambiguity aversion increases the price of reinsurance, which decreases demand. Moreover, the price of excess-loss reinsurance is greater than that of proportional reinsurance. Finally, when the insurer’s risk aversion is low or the reinsurer’s risk aversion is high, both the insurer and the reinsurer prefer the proportional reinsurance contract.}
}

@article{Cao2022,
title = {Stackelberg differential game for insurance under model ambiguity},
journal = {Insurance: Mathematics and Economics},
volume = {106},
pages = {128-145},
year = {2022},
%issn = {0167-6687},
doi = {10.1016/j.insmatheco.2022.06.003},
%url = {https://www.sciencedirect.com/science/article/pii/S0167668722000713},
author = {Jingyi Cao and Dongchen Li and Virginia R. Young and Bin Zou}
}

@article{Cao2023treeVchain,
title = {Reinsurance games with two reinsurers: Tree versus chain},
journal = {European Journal of Operational Research},
volume = {310},
number = {2},
pages = {928-941},
year = {2023},
issn = {0377-2217},
doi = {10.1016/j.ejor.2023.04.005},
%url = {https://www.sciencedirect.com/science/article/pii/S0377221723002746},
author = {Jingyi Cao and Dongchen Li and Virginia R. Young and Bin Zou}}

@article{Gu2020,
author = {Ailing Gu and Frederi G. Viens and Yang Shen},
title = {Optimal excess-of-loss reinsurance contract with ambiguity aversion in the principal-agent model},
journal = {Scandinavian Actuarial Journal},
volume = {2020},
number = {4},
pages = {342-375},
year  = {2020},
publisher = {Taylor & Francis},
doi = {10.1080/03461238.2019.1669218}
}

@article{Cao2022generaldivergence,
author = {Jingyi Cao and Dongchen Li and Virginia R. Young and Bin Zou},
title = {Stackelberg differential game for insurance under model ambiguity: general divergence},
journal = {Scandinavian Actuarial Journal},
volume = {2023},
number = {7},
pages = {735-763},
year  = {2023},
publisher = {Taylor & Francis},
doi = {10.1080/03461238.2022.2145233}
}

@article{Bai2022,
title = {A hybrid stochastic differential reinsurance and investment game with bounded memory},
journal = {European Journal of Operational Research},
volume = {296},
number = {2},
pages = {717-737},
year = {2022},
issn = {0377-2217},
doi = {10.1016/j.ejor.2021.04.046},
% url = {https://www.sciencedirect.com/science/article/pii/S0377221721003817},
author = {Yanfei Bai and Zhongbao Zhou and Helu Xiao and Rui Gao and Feimin Zhong},
keywords = {Decision analysis, Stochastic differential games, Reinsurance contract design, Investment, Delay}}

@article{BoonenGhossoub2019,
title = {On the existence of a representative reinsurer under heterogeneous beliefs},
journal = {Insurance: Mathematics and Economics},
volume = {88},
pages = {209-225},
year = {2019},
issn = {0167-6687},
doi = {10.1016/j.insmatheco.2019.07.008},
%url = {https://www.sciencedirect.com/science/article/pii/S0167668718305365},
author = {Tim J. Boonen and Mario Ghossoub},
keywords = {Optimal reinsurance design, Heterogeneous beliefs, Multiple reinsurers, Representative reinsurer, Deductible}
}

@article{BoonenGhossoub2021,
title = {Optimal reinsurance with multiple reinsurers: Distortion risk measures, distortion premium principles, and heterogeneous beliefs},
journal = {Insurance: Mathematics and Economics},
volume = {101},
pages = {23-37},
year = {2021},
note = {Behavioral Insurance: Mathematics and Economics},
issn = {0167-6687},
doi = {10.1016/j.insmatheco.2020.06.008},
author = {Tim J. Boonen and Mario Ghossoub}}

@article{boonen2016role,
  title={The role of a representative reinsurer in optimal reinsurance},
  author={Boonen, Tim J. and Tan, Ken Seng and Zhuang, Sheng Chao},
  journal={Insurance: Mathematics and Economics},
  volume={70},
  pages={196--204},
  year={2016},
  doi = {10.1016/j.insmatheco.2016.06.001},
  publisher={Elsevier}
}

@article{borch1960reciprocal,
  title={Reciprocal reinsurance treaties},
  author={Borch, Karl},
  journal={ASTIN Bulletin: The Journal of the IAA},
  volume={1},
  number={4},
  pages={170--191},
  year={1960},
  DOI={10.1017/S0515036100009557}, 
  publisher={Cambridge University Press}
}

@article{hojgaard1998optimal,
  title={Optimal proportional reinsurance policies for diffusion models},
  author={H{\o}jgaard, Bjarne and Taksar, Michael},
  journal={Scandinavian Actuarial Journal},
  volume={1998},
  number={2},
  pages={166--180},
  year={1998},
  publisher={Taylor \& Francis},
  doi = {10.1080/03461238.1998.10414000}
}

@article{schmidli2001optimal,
  title={Optimal proportional reinsurance policies in a dynamic setting},
  author={Schmidli, Hanspeter},
  journal={Scandinavian Actuarial Journal},
  volume={2001},
  number={1},
  pages={55--68},
  year={2001},
  doi = {10.1080/034612301750077338},
  publisher={Taylor \& Francis}
}

@article{schmidli2002minimizing,
  title={On minimizing the ruin probability by investment and reinsurance},
  author={Schmidli, Hanspeter},
  volume={12},
  number={3},
  pages={890--907},
  year={2002},
  publisher={Institute of Mathematical Statistics},
  doi = {10.1214/aoap/1031863173}
}

@article{Kroell2024,
title = {Stressing dynamic loss models},
journal = {Insurance: Mathematics and Economics},
volume = {114},
pages = {56-78},
year = {2024},
issn = {0167-6687},
doi = {10.1016/j.insmatheco.2023.11.002},
%url = {https://www.sciencedirect.com/science/article/pii/S0167668723000975},
author = {Emma Kroell and Silvana M. Pesenti and Sebastian Jaimungal},
keywords = {Reverse stress testing, Compound Poisson processes, KL divergence, Value-at-Risk, Conditional Value-at-Risk},
}

@article{Li2017,
title = {Optimality of excess-loss reinsurance under a mean–variance criterion},
journal = {Insurance: Mathematics and Economics},
volume = {75},
pages = {82-89},
year = {2017},
issn = {0167-6687},
doi = {10.1016/j.insmatheco.2017.05.001},
%url = {https://www.sciencedirect.com/science/article/pii/S0167668717301142},
author = {Danping Li and Dongchen Li and Virginia R. Young},
keywords = {Mean–variance criterion, Equilibrium reinsurance–investment strategy, Excess-loss reinsurance, Proportional reinsurance, Lévy insurance model}
}

@article{Cao_Li_Young_Zou_2023,
title={Reinsurance games with $\boldsymbol{{n}}$ variance-premium reinsurers: from tree to chain},
volume={53},
DOI={10.1017/asb.2023.24},
number={3},
journal={ASTIN Bulletin},
author={Cao, Jingyi and Li, Dongchen and Young, Virginia R. and Zou, Bin},
year={2023},
pages={706–728}}

\end{document}